%% file: sdphierarchies.tex
\documentclass[11pt]{article}

\def\full{1}

\def\showauthornotes{0}
\def\showtableofcontents{1}
\def\showkeys{0}
\def\showdraftbox{0}
\def\showcolorlinks{1}
\def\usemicrotype{1}
\def\showfixme{0}

\input{macros}

\let\pref=\prettyref

\newcommand{\rank}{\mathrm{rank}}

\newcommand{\Ins}{\Im}
\newcommand{\ia}{{ia}}
\newcommand{\jb}{{jb}}

\renewcommand{\Vsymb}{\mathrm{Var}}
\DeclareMathOperator*{\Cov}{Cov}

\newcommand{\maxtwocsp}{\problemmacro{Max 2-Csp}}

\title{%
  Rounding Semidefinite Programming Hierarchies via Global Correlation
}
\author{Boaz Barak\thanks{Microsoft Research New England, Cambridge MA.}
  \and Prasad Raghavendra\thanks{Georgia Institute of Technology, Atlanta
    GA.}
\and David Steurer\thanks{Microsoft Research New England, Cambridge MA.}}

\newcommand{\cp}{\mathsf{CP}}

\begin{document}

\maketitle
\draftbox
\thispagestyle{empty}

\begin{abstract}
We show a new way to round vector solutions of semidefinite programming
(SDP) hierarchies into integral solutions, based on a connection between
these hierarchies and the spectrum of the input graph. We demonstrate the
utility of our method by providing a new SDP-hierarchy based algorithm
for constraint satisfaction problems with 2-variable constraints (2-CSP's).

More concretely, we show for every $2$-CSP instance $\Ins$ a rounding algorithm
for $r$ rounds of the Lasserre SDP hierarchy for $\Ins$ that
obtains an integral solution that is at most $\e$ worse than the relaxation's value (normalized to lie in $[0,1]$), as
long as
\[
  r > k\cdot\rank_{\geq \theta}(\Ins)/\poly(\e) \;,
\]
where $k$ is the alphabet size of $\Ins$, $\theta=\poly(\e/k)$,
and $\rank_{\geq \theta}(\Ins)$ denotes the number of eigenvalues
larger than $\theta$ in the normalized adjacency matrix of the constraint graph of $\Ins$.

In the case that $\Ins$ is a \uniquegames instance, the threshold $\theta$ is
only a polynomial in $\e$, and is independent of the alphabet size. Also in
this case, we can give a non-trivial bound on the number of rounds for
\emph{every} instance. In particular our result yields an SDP-hierarchy based
algorithm that matches the performance of the recent subexponential
algorithm of Arora, Barak and Steurer (FOCS 2010) in the worst case, but runs faster on
a natural family of instances,
thus further restricting the set of possible hard instances for Khot's Unique Games Conjecture.

Our algorithm actually requires less than the $n^{O(r)}$ constraints
specified by the $r^{th}$ level of the Lasserre hierarchy, and in
some cases $r$ rounds of our program can be evaluated in time
$2^{O(r)}\poly(n)$.
\end{abstract}

\clearpage

\ifnum\showtableofcontents=1
{
\tableofcontents
\thispagestyle{empty}
 }
\fi

\clearpage
\setcounter{page}{1}

\newcommand{\rvar}{random variable\xspace}
\newcommand{\rvars}{random variables\xspace}

\newcommand{\rval}{real-valued\xspace}

\section{Introduction}

This paper is concerned with hierarchies of semi-definite programs (SDP's). Semidefinite programs are an extremely useful tool in
algorithms and in particular approximation algorithms~(e.g., \cite{GoemansW95,KargerMS98}). SDP's involve finding an integral
(say $0/1$) solution for some optimization problem, by using convex programming to find a fractional/high-dimensional solution
and then \emph{rounding} it into an integral solution. Sherali and Adams~\cite{SheraliA90}, \Lovasz and
Schrijver~\cite{LovaszS91}, and, later Lasserre~\cite{Lasserre01},  proposed systematic ways, known as \emph{hierarchies}, to
make this convex relaxation tighter, thus ensuring that the fractional solution is closer to an integral one. These hierarchies
are parameterized by a number $r$, called the \emph{level} or \emph{number of rounds} of the hierarchy. Given a program on $n$
variables, optimizing over the $r^{th}$ level of the hierarchy can be done in time $n^{O(r)}$. The gap between integral and
fractional solutions decreases with $r$, and reaches zero at the $n^{th}$ level, where the program is guaranteed to find an
optimal integral solution. The paper \cite{Laurent03} surveys and compares the different hierarchies proposed in the literature,
see also the recent survey \cite{ChlamtacT10}.

These semidefinite programming hierarchies have been of some interest in recent years, since they provide natural candidate algorithms for many computational problems. In particular, whenever the basic semidefinite or linear program provides a suboptimal approximation factor, it makes sense to ask how many rounds of the hierarchy are required to significantly improve upon this factor. Unfortunately, taking advantage of these hierarchies has often been difficult, and while some algorithms (e.g.,~\cite{AroraRV09}) can be encapsulated in, say, level $3$ or $4$ of some hierarchies, there have been relatively few results (e.g. ~\cite{Chlamtac07,BhaskaraCCFV10}) that use higher levels to obtain new algorithmic results. In fact, there has been more success in showing that high levels of hierarchies \emph{do not} help for many computational problems~\cite{AroraBLT06,SchoenebeckTT07a,GeorgiouMPT07,RaghavendraS09c,KhotS09}. In particular for 3SAT and several other NP-hard problems, it is known that it takes $\Omega(n)$ rounds of the strongest SDP hierarchy (i.e., Lasserre) to improve upon the approximation rate achieved by the basic SDP (or sometimes even simpler algorithms)~\cite{Schoenebeck08,Tulsiani09}.

Semidefinite hierarchies are of particular interest in the case of problems related to Khot's \emph{Unique Games Conjecture} (UGC)~\cite{Khot02a}. Several works have shown that for a wide variety of problems, the UGC implies that (unless $\p=\np$) the basic semidefinite program cannot be improved upon by any polynomial-time algorithm~\cite{KhotKMO04,MosselOO05,Raghavendra08}. Thus in particular the UGC predicts that for all these problems, it will take a super-constant (and in fact polynomial, under widely believed assumptions) number of hierarchy rounds to improve upon the basic SDP. Investigating this prediction, particularly for the \uniquegames problem itself and other related problems such as \maxcut, \sparsestcut and \smallsetexpansion, has been the focus of several works, and it is known that at least $(\log\log n)^{\Omega(1)}$ rounds are required for a non-trivial approximation~\cite{RaghavendraS09c,KhotS09} by a natural (though not strongest possible) SDP hierarchy. However, no non-trivial upper bound was known prior to the current work, and so it was conceivable that these lower bounds can be improved to $\Omega(n)$.

Recently, Arora, Barak and Steurer~\cite{AroraBS10} gave a $2^{n^{\poly(\e)}}$-time algorithm for solving the \uniquegames and \smallsetexpansion problems (where $\e$ is the completeness parameter, see below). However, their algorithm did not use semidefinite programming hierarchies, and so does not immediately imply an upper bound on the number of rounds needed.

\subsection{Our results}

Our main contribution is a new method to analyze and round SDP hierarchies. We elaborate more on our method in Section~\ref{sec:technique}, but its high level description is that uses \emph{global correlations} inside the high-dimensional SDP solution, combined with the hierarchy constraints, to obtain a better rounding of this solution into an integral one. We believe this method can be of general utility, and in particular we use it here to give new algorithms for approximating constraint satisfaction problem on two-variable constraints (2-CSP's), that run faster than the previously known algorithms for a natural family of instances. To state our results we need the notion of a \emph{threshold rank}.

\paragraph{Threshold rank of graphs and 2-CSPs}  The \emph{$\tau$-threshold rank} of a
regular graph $G$, denoted  $\rank_{\geq \tau}(G)$, is the number of eigenvalues of the
normalized adjacency matrix of $G$ that are larger than $\tau$.\footnote{In this paper we only consider regular undirected graphs, although we allow non-negative weights and/or parallel edges. Every such graph can be identified with its normalized adjacency matrix, whose $(i,j)^{th}$ entry is proportional to the weight of the edge $(i,j)$, with all row and column sums equalling one. Similarly, we restrict our attentions to 2-CSP's whose constraint graphs are regular.  However, our definitions and results can be appropriately generalized for non-regular graphs and 2-CSPs as well.}   An instance $\Ins$ of a \maxtwocsp problem consists of a  regular  graph $G_{\Ins}$, known as the \emph{constraint graph} of $\Ins$ over a vertex set $[n]$, where every edge $(i,j)$ in the graph is labeled with a relation $\Pi_{i,j} \subseteq [k] \times [k]$ ($k$ is known as the \emph{alphabet size} of $\Ins$). The \emph{value} of an assignment $x \in [k]^n$ to the variables of $\Ins$, denoted $\val_{\Ins}(x)$, is equal to the probability that $(x_i,x_j) \in \Pi_{i,j}$ where $(i,j)$ is a random edge in $G_{\Ins}$.  The \emph{objective value} of $\Ins$ is the maximum $\val_{\Ins}(x)$ over all assignments $x$. We say that $\Ins$ is \emph{$c$-satisfiable} if $\Ins$'s objective value is at least $c$. We define $\rank_{\geq \tau}(\Ins) = \rank_{\ge \tau}(G_{\Ins})$. Our main result is the following:

\begin{theorem} \label{thm:main}
There is a constant $c$ such that for every $\e > 0$, and every \maxtwocsp instance $\Ins$ with objective value
  $v$, alphabet size $k$ the following holds:
the objective value $\mathrm{sdpopt}(\Ins)$ of the $r$-round Lasserre hierarchy for $ r \geq k  \cdot \rank_{\geq \tau}(\Ins) / \e^c$ is within $\epsilon$ of the objective value $v$ of $\Ins$, i.e.,
$  \mathrm{sdpopt}(\Ins) \leq v +  \eps.$
Moreover, there exists a polynomial time rounding scheme that finds an assignment $x$ satisfying $\val_{\Ins}(x) > v - \e$ given optimal SDP solution as input.
\end{theorem}

\paragraph{Results for \uniquegames constraints} We say that a \maxtwocsp instance is a \uniquegames instance if all the relation $\Pi_{i,j}$ have the form that $(a,b) \in \Pi_{i,j}$ iff $a = \pi_{i,j}(b)$ where $\pi_{i,j}$ is a permutation of $[k]$.  As mentioned above, the performance of SDP hierarchy on \uniquegames instances and related problems is of particular interest. We obtain somewhat stronger quantitative results for \uniquegames instances. Also, as remarked below, our results are ``morally stronger'' in this case, since it's conceivable that the hardest instances for these types of problems have small threshold rank.  First, we show that for \uniquegames instances the threshold $\tau$ in  \pref{thm:main} does not need to depend on the alphabet size. Namely, we prove

\begin{theorem} \label{thm:main-ug} There is an  algorithm, based on rounding $r$ rounds of the Lasserre hierarchy and a constant $c$, such that for every $\e>0$  and input \uniquegames instance $\Ins$ with objective value $v$, alphabet size $k$, satisfying $\rank_{\geq \tau}(\Ins) \leq \e^c r / k$, where $\tau = \e^c$, the algorithm outputs an assignment $x$ satisfying $\val_{\Ins}(x) > v - \e$.
\end{theorem}

The Unique Games Conjecture is about a specific approximation regime for
\uniquegames.
Given a \uniquegames instance with optimal value $1-\e$, the goal is to
find an assignment with value at least $1/2$.

We also show that in this case a sublinear (and in fact a small root)
number of rounds suffice to get such an
approximation in the worst case, regardless of the instance's threshold
rank.
Moreover, we also show that such an approximation can be obtained in a
number of rounds that depends on the $\tau$-threshold rank for $\tau$ that
is close to $1$ (as opposed to the small value of $\tau$ needed for
Theorems~\ref{thm:main} and~\ref{thm:main-ug}).

\begin{theorem} \label{thm:subexp} There is an algorithm, based on rounding
  $r$ rounds of the Lasserre hierarchy and a constant $c$, such that for
  every $\e>0$ and input \uniquegames instance $\Ins$ with objective value
  $1-\e$ and alphabet size $k$, satisfying $r \geq c k\cdot \min \{
  n^{c\e^{1/3}} , \rank_{\geq 1-c\e}(\Ins) \}$, the algorithm outputs an
  assignment $x$ satisfying $\val_{\Ins}(x) > 1/2$.
\end{theorem}

\paragraph{Examples of graphs with small threshold rank}
Many interesting graph families have small $\tau$-threshold rank for some
small constant $\tau$.
Random degree $d$ graphs have $\tau$-threshold equal to $1$ for any $\tau >
c/\sqrt{d}$.
More generally, graphs where small subsets of vertices have bounded
edge-expansion, referred to as {\it small-set expanders}, also have small
threshold rank.
For instance, if every set of size $o(n)$ expands by at least
$poly(\epsilon)$ in a graph $G$, then $\rank_{1-\epsilon}(G)$ is at most
$n^{poly(\epsilon)}$ \cite{AroraBS10}.
Generalizing this result, \cite{Steurer10c} showed that if in a graph $G$
every set of size $o(n)$ vertices has near-perfect expansion, then it
implies upper bounds on $\rank_{\tau}(G)$ for $\tau$ close to $0$.

Also, as noted in~\cite{AroraBS10}, \emph{hypercontractive} graphs
(i.e., graphs whose $2$ to $4$ operator norm is bounded) have at most polylogarithmic $\tau$-threshold rank for every constant
$\tau>0$. For several 2-CSP's  such as  \maxcut, \uniquegames, \smallsetexpansion,  \sparsestcut, the constraint graphs for the
canonical ``problematic instances'' (i.e., integrality gap examples~\cite{FeigeS02b,KhotV05,KhotS09,RaghavendraS09c}) are all
hypercontractive, since they are based on either the noisy Gaussian graph or noisy Boolean cube. In fact, it is conceivable that
the \smallsetexpansion problem is trivial on graphs with large threshold rank, in the sense that we do not know of any example of
an instance having, say, $\log^{\omega(1)} n$ $0.99$-threshold rank, and objective value smaller than $1/2$. (For the
$\uniquegames$ and $\maxcut$ problems it is trivial to construct instances with large threshold rank by taking many disjoint
copies of the same instance, though it could still be the case that the hardest instances are the ones with small threshold
rank.) On the other hand, for other 2-CSPs such as \labelcover , some natural hard instances have linear threshold-rank. For
example this is the case if one considers the natural ``clause vs. variable'' or ``clause vs. clause'' 2-CSP obtained from random
instances of 3SAT (which is not surprising given that a non-trivial approximation for random 3SAT requires $\Omega(n)$ levels of
the Lasserre hierarchy~\cite{Schoenebeck08}).

\paragraph{Algorithm efficiency} Our algorithm actually does not require the full power of the Lasserre hierarchy. First, we can use the relaxed variant with
\emph{approximate} constraints studied in~\cite{KhotS09,RaghavendraS09c,KhotPS10}. Second, in the proof of \pref{thm:subexp}, we don't need to utilize the
constraints on all $\binom{n}{r}$ $r$-sized subsets of $n$ variables, but rather sufficiently many random sets suffice. As a
result, we can implement our $r$-round algorithm in time $2^{O(r)}\poly(n)$.


\subsection{Related works}

\paragraph{Subspace enumeration algorithms}
For \uniquegames and related problems, previous works~\cite{KollaT07,Kolla10,AroraBS10} used \emph{subspace enumeration} to give algorithms with similar running
time to \pref{thm:subexp} in the case that the threshold rank of the \emph{label
  extended graph} of the instance is small. This is known to be a stronger
condition on the instances than bounding the threshold rank of the constraint graph. The only known bound on the $1-\e$ threshold rank
of the label extended graph in terms of the $1-\e$ threshold rank of the constraint graph loses a factor of about $n^{\e}$
\cite{AroraBS10}. These subspace enumeration algorithms also only applied to \emph{nearly satisfiable} instances (whose objective value is close to $1$),
and so did not give guarantees comparable to Theorems~\ref{thm:main} and~\ref{thm:main-ug}. As mentioned below in
\pref{sec:technique}, SDP-based algorithms have some robustness advantages over spectral techniques. SDP hierarchies are
also easily shown to yield polynomial-time approximation scheme for 2CSPs whose constraint graphs can have very \emph{high} threshold rank such as
bounded tree width graphs and regular planar graphs (or more generally any \emph{hyperfinite} family of graphs, see e.g. \cite{HassidimKNO09} and the references therein).

\paragraph{Approximation schemes for (pseudo) dense CSP's}
For general 2CSP's, several works gave polynomial-time approximation schemes for \emph{dense} and \emph{pseudo-dense}
instances~\cite{FriezeK99,AlonCHKRS10,Coja-OghlanCF10}. Our work generalizes these results, since pseudo-density is a stronger
condition than having a constraint graph of low threshold rank.   Furthermore, for an $\epsilon$-approximation the degree of the instance needed by these works is exponential in $\frac{1}{\epsilon}$, while the results of this work apply even on random graphs of degree $poly(1/\epsilon)$.

\paragraph{Analyzing SDP hierarchy}
Using very different methods, Chlamtac~\cite{Chlamtac07} and Bhaskara et al~\cite{BhaskaraCCFV10} gave LP/SDP-hierarchy based
algorithms for hypergraph coloring and the densest subgraph problem respectively. As mentioned above, several works gave \emph{lower bounds} for LP/SDP hierarchies. In particular \cite{RaghavendraS09c,KhotS09} showed that approximation such as those achieved in Theorem~\ref{thm:subexp}  for \uniquegames problem require $\log\log ^{\Omega(1)} n$ rounds of a relaxed variant of the Lasserre hierarchy. This relaxed variant captures our hierarchy as well. Schoenebeck~\cite{Schoenebeck08} proved that achieving a non-trivial approximation for 3SAT on random instances requires $\Omega(n)$ rounds in the Lasserre hierarchy, while Tulsiani~\cite{Tulsiani09} showed that Lasserre lower bounds are preserved under common types of NP-hardness reductions.

In a concurrent and independent work, Guruswami and
Sinop~\cite{GuruswamiS11b} gave results very similar to ours. They also use
the Lasserre hierarchy to get an approximation scheme with similar
performance to our \pref{thm:main} for 2-CSPs, and in fact even consider
generalizations involving additional (approximate) global  linear
constraints. They also get essentially the same results for \uniquegames as
our \pref{thm:subexp}. 
Furthermore, their rounding algorithm is the same as ours.
However, there are some differences both in results and the proof. First, although \cite{GuruswamiS11b} use a notion similar to our local-to-global correlation, they view it differently, and interestingly relate it to the problem of column selection for low rank approximations of matrices. Also, apart from the special case of unique constraints, they work with the threshold rank of the \emph{label extended graph}, as opposed to the constraint graph as is the case here (however for binary alphabet these two graphs coincide). Their analysis relies on the full power of the Lasserre hierarchy, whereas we show that a weaker hierarchy is sufficient in the \uniquegames case, and can even be done faster (i.e., $\exp(r)\poly(n)$ vs $n^{O(r)}$).

\section{Our techniques} \label{sec:technique}

We now describe, on a very high and imprecise level, the ideas behind our rounding algorithm and its analysis.  A semidefinite programming relaxation of an optimization problem yields a set of vectors $v_1,\ldots,v_n$ satisfying certain conditions and achieving some objective value $c$. The goal of a rounding algorithm is to transform this set of vectors into, say, a $+1/-1$ solution, satisfying the same conditions and achieving value $c'$ that is close to $c$. At a very high level, our main result is  that if these vectors have some non-trivial \emph{global correlation}, then a good rounding can be achieved with a non-trivially small number of hierarchy rounds. Our second observation is that in several cases, the vectors corresponding to a \emph{good} SDP solution can be shown to have significant mass inside some low-dimensional subspace, and that implies a lower bound on their global correlation. Below we elaborate on what we mean, using the $\maxcut$ problem (which is a special case of \uniquegames ) as an illustrative example. Our result for \maxcut is worked out in more detail in Section~\ref{sec:warmup}.

\nnspace\paragraph{Rounding SDP's using a small basis}
The SDP solution for $\maxcut$ problem consists of a sequence $\cV = v_1,\ldots,v_n$ of unit vectors, and the objective value is the expectation of $(1-\iprod{v_i,v_j})/2$ over all edges $\{i,j\}$ in the input graph. Note that in the case that the vectors $v_1,\ldots,v_n$ are one dimensional unit vectors (i.e., $v_i \in \{ \pm 1 \}$), $\cV$ exactly corresponds to a cut in the graph, and the objective value measures the fraction of edges cut. Now, suppose that you could find $r$ vectors $v_{i_1},\ldots,v_{i_r} \in \cV$, whom we'll call the \emph{basis vectors}, such that every other $v \in \cV$ has some significant projection $\rho$ into the span of $v_{i_1},\ldots,v_{i_r}$. That is, if we let $P$ be the projection operator corresponding to this space, then for every $v\in\cV$ , $\normt{Pv} \geq \rho$.  It turns out that in this case, if  $\rho$ is sufficiently close to $1$ and the vector solution $\cV$ satisfied $r+2$ rounds of an appropriate SDP hierarchy, then we can round $\cV$ to achieve a very good cut. The intuition behind this is the following:  the constraints of $r+2$ hierarchy rounds allow us to essentially assume without loss of generality that the vectors $v_{i_1},\ldots,v_{i_r}$ are one-dimensional. That is, after applying an appropriate rotation, we can think of each one of them as a vector of the form $(\pm 1, 0, \ldots, 0)$. Moreover, our assumption implies that every other vector in $v$ has a magnitude of at least $\rho$ in its first coordinate. Now one can show that simply rounding each vector to the sign of its first coordinate will result in a $\pm 1$ assignment to the vertices corresponding to a good cut.

\nnspace\paragraph{Local to global correlation} From the above discussion, our goal of rounding SDP hierarchies is reduced to finding a small number of basis vectors $v_{i_1},\ldots,v_{i_r}$ such that every (or at least most) other vector in the solution $\cV$ has very large projection into their span. But, why should such vectors exist? We show that we can assume they exist if the original \maxcut instance has small \emph{threshold rank}.  The latter is a condition that, as mentioned above, holds for many natural families of instances, including the canonical ``hard instances'' that are known to fool the GW algorithm--- the noisy sphere and noisy Gaussian graphs~\cite{FeigeS02b,RaghavendraS09c}. The key concept behind our proof is the notion of \emph{local vs global correlations}. It is a very well known
property of expander graphs that random edges behave similarly to pairs of independently chosen vertices with respect to some tests. Specifically, if
$G$ is an $n$-vertex expander in the sense that the normalized adjacency matrix $A_G$'s second largest
eigenvalue is at most $\e$, and $f$ is a bounded function mapping vertices
to numbers, then we know that $\E_{i,j}[ |f(i)-f(j)|^2] \in (1\pm O(\e))\E_{i\sim
  j}[ |f(i)-f(j)|^2])$, where the former expectation is over pairs of
vertices and the latter is over pairs connected by an edge. In other words,
expander graphs imply that if $f$ is \emph{locally correlated} over the
edges of an expander graph, then it is also \emph{globally correlated}. In
fact, this is easily shown to hold even if $f$ maps vertices not into
numbers but into vectors--- i.e., if $v_1,\ldots,v_n$ are unit vectors that
are locally correlated over the edges of $G$ then they are also globally
correlated.  Indeed, this property of expanders has been used in the work of \cite{AroraKKSTV08}, who showed that the basic SDP program for \uniquegames can be successfully rounded if
the input graph is an expander.

Our starting point is to observe that a somewhat similar, though much weaker condition holds even when the graph has at most, say, $r/100$ eigenvalues larger than $\e$. In this case it's possible to show that,say, if \textbf{(*)} $\E_{i\sim j} \iprod{v_i,v_j}^2 > 100\sqrt{e}$, then \textbf{(**)} $\E_{i,j} \iprod{v_i,v_j}^2 \ge 1/r$ (see Section~\ref{sec:local-global-lowrank}). If $r$ is super-constant, the condition \textbf{(**)} does not seem a-priori useful for obtaining a good integral solution. Indeed, the standard integrality gap example of \maxcut is a graph with fairly small (polylogarithmic) number of large eigenvalues, but no good integral solution. However, \textbf{(**)} does imply that we can find at least one vector $v_{i_1}$ such that $\E_j \iprod{v_{i_1},v_j}^2 \geq 1/r$. We can now replace each vector $v \in \cV$ with its projection into the orthogonal space to $v_{i_1}$ and continue until we either get stuck or find a basis $v_{i_1},\ldots,v_{i_r}$ such that (almost all) vectors $v\in\cV$ have most of their mass in $\mathrm{Span}\{ v_{i_1},\ldots, v_{i_r} \}$, in which case we can successfully round the solution. The only way we can get stuck is if at some point we get that  \textbf{(*)} is violated. Now, in the case of \maxcut, if \textbf{(*)} was violated initially, then the value of the SDP would be about $1/2$, which is trivial to round by just taking a random cut. To show that we can easily round even when \textbf{(*)} is violated at some later point in the process, it's useful to switch to the distribution view of SDP hierarchies.

\nnspace\paragraph{Distribution view of SDP's} Another, often beneficial way to view SDP hierarchies is as providing \emph{distribution} on integral solutions (see Section~\ref{sec:local-distributions}). In this view, for every set of $r+2$ vertices $i_1,\ldots,i_r,i_{r+1},i_{r+2}$, the SDP hierarchy provides a distribution $X_{i_1},\ldots,X_{i_r}$ over $\pm 1$. Moreover, we require that distributions on overlapping sets will be consistent, and that the for every two variables $i,j$ the expectation $E[X_iX_j]$ will equal the inner product $\iprod{v_i,v_j}$ of the corresponding vectors. The challenge in rounding the SDP is that there is not necessarily a way to sample simultaneously the random variables $X_1,\ldots,X_n$ in some consistent way. The projection of a vector $v$ into the span of $v_{i_1},\ldots,v_{i_r}$ turns out to capture (an appropriate notion of) the \emph{mutual information} between the variable $X_{i_1}$ and the variables $X_{i_1},\ldots,X_{i_r}$. Looked at from this viewpoint, our rounding algorithm involves choosing an assignment from the  distribution for the basis vertices, and \emph{conditioning} on its value. As long as \textbf{(**)} holds, we can find a random variable $X_i$ such that conditioning on $X_i$ will significantly decrease the entropy of the remaining variables. When we get stuck and \textbf{(*)} is violated, it means that for a typical edge $i\sim j$, the random variables $X_i$ and $X_j$ are close to being \emph{statistically independent}. This means that just sampling each $X_i$ independently will give approximately the same value on a typical constraint.

\nnspace\paragraph{Threshold rank vs global correlation} Whenever the graph has small number of large eigenvalues, the condition
that local correlation implies global correlation holds. This is useful to simulate eigenspace enumeration algorithms such as
used by \cite{KollaT07,Kolla10,AroraBS10,Steurer10c} since in the case of \uniquegames (and other related problems),  a good SDP
solution must be locally well correlated. But the notion of local to global correlation is somewhat more general and robust than
having small threshold rank. For example, adding $\sqrt{n}$ isolated vertices to a graph will increase correspondingly the number
of eigenvectors with value $1$, but will actually not change by much the local to global correlation. This captures to a certain
extent the fact that SDP-based solutions are more robust than the spectral based algorithms. (A similar example of this
phenomenon is that adding a tiny bipartite disjoint graph to the input graph makes the smallest eigenvalue become $-1$, but does
not change by much the value of the Goemans-Williamson SDP.) We hope that this robustness of the SDP-based approach will enable
further improvements in the future.

\begin{remark} \pref{thm:subexp} considers a different parameter than Theorems~\ref{thm:main} and~\ref{thm:main-ug}. 
The latter two results consider threshold ranks for a small (i.e., close to $0$) threshold $\tau$, and achieve a very good approximation.
In contrast, \pref{thm:main} considers threshold $\tau$ that is close to $1$, but only achieve a rough approximation (corresponding to the approximation
guarantee relevant to the unique games conjecture). This is also manifested in some technical differences in the proofs.
\end{remark}

\section*{Organization}

We begin by fixing notation and a few formal definitions in the next section.  For the purpose of exposition, we first present an
algorithm for \maxcut on low-rank graphs using the Lasserre hierarchy in \pref{sec:warmup}. Following this, the general algorithm
for 2-CSPs on low-rank graphs is presented in \pref{sec:general-2-csp}.  The connection between local and global correlations in
low-rank graphs that is central to our algorithms, is outlined in \pref{sec:local-global-lowrank}.  To implement our general
approach in a hierarchy weaker than Lasserre hierarchy, we outline an argument to obtain low-rank approximation to any set of
vectors in \pref{sec:low-rank-approx}.   The final section (\pref{sec:round-uniq-games}) of the paper is devoted to
subexponential time algorithm for Unique Games.

\section{Preliminaries}

\ifnum\full=0

We will use capital letters $X,Y$ to denote random
variables, and lower-case letters to denote assignments to these
random variables. $\cp(X)$ denotes the collision probability of a random variable.
If $X$ is a random variable and $a$ is an element in $X$'s domain, we use $X_a$ to denote the indicator variable
that equals $1$ if $X=a$ and equals $0$ otherwise. For a random variable $X$ whose range is $[k]$, we define the \emph{variance} of $X$ as
\begin{equation}
 \Var[X] \defeq \sum_{a \in [k]} \Var[X_{1a}]  = 1 - \cp(X) \;, \label{eq:defvarianceoverk}
 \end{equation}
where $\cp(X)$ denotes the collision probability of $X$.

\else

We will use capital letters $X,Y$ to denote random
variables, and lower-case letters to denote assignments to these
random variables.

For a real-valued random variable $X$, let $\Var[X]$ denote its
variance.  In this work, we will use random variables taking values over a
range $[k] = \{1\ldots k\}$.  For a random variable $X$ over $[k]$,
and $a \in [k]$, let $X_{a}$ denote the indicator of the event that
$X = a$.  We define the variance of $X$ to be,
\begin{equation}
 \Var[X] \defeq \sum_{a \in [k]} \Var[X_{a}]  = 1 -\cp(X) \mper \label{eq:defvarianceoverk}
 \end{equation}
The \emph{collision probability} of $X$  is defined as
\begin{displaymath}
  \cp(X)
  \defeq \Prob[{X}{X'}]{X= X'}
  \mcom
\end{displaymath}
where $X'$ is an independent copy of $X$ (so that the sequence
$X,X'$ is \iid).
It is easy to see that the variance and collision probability are
related by,
$$ \cp(X) = 1 - \Var[X] \mper $$

For two jointly-distributed random variables $X,Y$, let $\brac{X\mid Y=y}$
denote the random variable $X$ conditioned on the event that $Y = y$.
If it is clear from the context, we write $(X|y)$ for $(X|Y=y)$.  We will
denote by $\E_{\set{Y}}\Var[X|Y]$ the following quantity,
$$ \E_{\set{Y}}\Var[X|Y] = \E_{y}\left[ \Var[(X|Y=y)] \right] \mper$$

\fi

\ifnum\full=1 \subsection{Unique Games}

\begin{definition}
\else \nnspace\paragraph{Unique games}
\fi
	An instance of \uniquegames consists of a graph $G =
	(V,E)$, a label set $[k] = \{1,\ldots, k\}$ and a bijection
	$\pi_{ij} \from [k] \to [k]$ for every edge $(i,j) \in E$.
	A labelling $\ell \from V \to [k]$ is said to {\it satisfy}
	an edge $(i,j)$ if $\pi_{ij}(\ell(i)) = \ell(j)$.
	The goal is to find a labeling $\ell \from V \to [k]$ that
	satisfies the maximum number of edges namely,
	\beq \text{maximize} \,\,\,\,\,  \Pr_{(i,j) \in E}\left\{ \pi_{ij}(\ell(i)) = \ell(j)
	\right\} \eeq
\ifnum\full=1
\end{definition}
\fi

\ifnum\full=1
\subsection{Local Distributions} \label{sec:local-distributions}
\else
\nnspace\paragraph{Local distributions.}
\fi
Let $V=[n]$ be a set of vertices and let $[k]$ be a set of labels. An \emph{$m$-local distribution} is a distribution
$\mu^T$ over the set of assignments $[k]^T$ of the vertices  of some set $T\subseteq V$ of size at most $m+2$. (The choice of $m+2$ is immaterial but will be convenient later on.)
A collection of $m$-local distributions $\set{\mu_T}_{T\sse V,~\card T \le
  m+2}$ is \emph{consistent} if for all $T,T'\sse V$ with $\card{T},\card
{T'}\le m+2$, the distributions $\mu_T$ and $\mu_{T'}$ are consistent on
their intersection $T\cap T'$. We sometimes will view these distributions as random variables, hence writing $X_i^{(T)}$ for the random variable over $[k]$ that is distributed according to the label that  $\mu^{T\cup \{i\}}$ assigns to $i$, and refer to a collection $X_1,\ldots,X_n$ of \emph{$m$-local random variables}. However, we stress that these are not necessarily jointly distributed random variables, but rather for any subset of at most $m+2$ of them, one can find a sample space on which they are jointly distributed.  For succinctness, we omit the superscript for variables
$\cramped{\super{\cramped{X_i}} S}$ whenever it is clear from the context. For example, we will use  $\set{X_i \mid X_S}$ is short for the random variable obtained by
  conditioning $\super X {S\cup \set i}_i$ on the variables $\set {\super X
    {S\cup \set i}_j}_{j\in S}$;\footnote{Strictly speaking, the range of the random variable $\set{X_i\mid X_S}$
  are random variables with range $[k]$. For every possible value $x_S$ for $X_S$, one obtains a $[k]$-valued
  random variable $\set{X_i\mid X_S=x_S}$.} and use $\Prob{X_i=X_j \mid X_S}$ is short for the $[0,1]$-valued random
  variable $\Prob{\super X{S\cup\set{i,j}}_i=\super X{S\cup\set{i,j}} _j \mid \super X {S\cup\set{i,j}}_S}$.

\ifnum\full=1\subsection{Lasserre Hierarchy} \else \nnspace\paragraph{Lasserre hierarchy} \fi
Let $U$ be a \uniquegames instance with constraint graph $G=(V,E)$, label set
$[k]=\set{1,\ldots,k}$, and bisections $\set{\pi_\ij}_{ij\in E}$.
An \emph{$m$-round Lasserre solution} consists of $m$-local random
variables $X_1,\ldots,X_n$ and vectors $v_{S,\alpha}$ for all vertex sets $S \sse V$
with $\card S \le m+2$ and all local assignments $\alpha\in [k]^S$.
A Lasserre solution is \emph{feasible} if the local random variables are
consistent with the vectors, in the sense that for all $S,T\sse V$ and
$\alpha\in [k]^S,\beta\in [k]^T$ with $\card{S\cup T}\le m+2$, we have
\beq
  \iprod{v_{S,\alpha},v_{T,\beta}} = \Prob{X_S=\alpha,~X_T=\beta}\mper
\eeq
The objective is to maximize the following expression
\beq
  \E_{ij\in E} \Prob{X_j=\pi_\ij(X_i)}\mper
\eeq
An important consequence of the existence of the vectors $v_{S,\alpha}$ is
that for every set $S\sse V$ with $\card{S}\le m$ and local assignment
$x_S\in [k]^S$, the matrix $\Set{\Cov(X_\ia,X_\jb\mid X_S=x_S)}_{i,j\in
  V,\,a,b\in[k]}$ is positive semidefinite.

\section{Warmup -- MaxCut Example} \label{sec:warmup}

For the sake of exposition, we first present an algorithm for the \maxcut problem on low-rank graphs.  In the \maxcut problem, the input consists of a graph $G = (V,E)$ and the goal is to find a cut $S \cup \bar S = V$ of the vertices that maximizes the number of edges crossing, i.e., maximizes $|E(S,\bar S)|$.

The Goemans-Williamson SDP relaxation for the problem assigns a unit vector $v_i$ for every vertex $i \in V$, so as to maximize the average squared length $E_{i,j \in E} \norm{v_i - v_j}^2$ of the edges.  Formally, the SDP relaxation is given by,
$$ \textrm{maximize }    \E_{i,j \in E} \norm{v_i - v_j}^2  \ \textrm{ subject to } \norm{v}_i^2 = 1 ~\forall i \in V$$

Stronger SDP relaxations produced by hierarchies such as Sherali-Adams and Lasserre hierarchy also yield probability distributions over local assignments.

More precisely, given a $m$-round Lasserre SDP solution, it can be associated with a set of $m$-local random variables $X_1, \ldots, X_n$ taking values in $\set{-1,1}$.  For an edge $(i,j)$, its contribution to the SDP objective value ($\norm{v_i - v_j}^2$) is equal to the probability that the edge $(i,j)$ is cut under the distribution of local assignments $\mu_{ij}$, namely,
$$  \Pr_{\mu_{ij}}[X_i \neq X_j] = \norm{v_i - v_j}^2 \mper$$

Consequently, in order to obtain a cut with value {\it close} to the SDP objective, it is sufficient to jointly sample $X_1, \ldots, X_n$, such that on every edge $(i,j)$ the distribution of $X_i$ and $X_j$ is {\it close} to the corresponding local distribution $\mu_{ij}$.  However, the variables $X_1, \ldots, X_n$ are not jointly distributed, and hence cannot all be sampled together.

As a first attempt, let us suppose we sample each $X_i$ {\it independently} from its associated marginal $\mu_i$.  If on most edges $(i,j)$, the distribution of the resulting samples $X_i,X_j$ is close to $\mu_{ij}$, then we are done.
On an edge $(i,j)$, the local distribution $\mu_{ij}$ is {\it far} from the independent sampling distribution $\mu_i \times \mu_j$  only if the random variables $X_i,X_j$ are {\it correlated}.  Henceforth, these correlations across the edges would be refered to as ``local correlations".  A natural measure for  correlations that we will utilize here is defined as $ \Cov(X_i,X_j) = \E[X_i X_j] - \E[X_i]\E[X_j] $.  Using this measure, the statistical distance between independent sampling ($\mu_i \times \mu_j$) and correlated sampling ($\mu_{ij}$) is given by $$ \norm{\mu_{ij}- \mu_{i} \times \mu_j}_1 \leq \abs{\Cov(X_i,X_j)} \mper$$
(See \pref{lem:stat-dist-correlation} for a  more general version of the above bound).

On the flip side, the existence of correlations makes the problem of sampling $X_1,\ldots,X_n$ easier!  If two variables $X_i,X_j$ are correlated, then sampling/fixing the value of $X_i$ reduces the uncertainty in the value of $X_j$.  More precisely, conditioning on the value of $X_i$ reduces the variance of $X_j$  as shown below:
$$ \E_{\{X_i\}}\Var[X_j | X_i] = \Var[X_j] - \frac{1}{\Var[X_i]} \left[\Cov(X_i,X_j)\right]^2 \mper $$
Therefore, if we pick an $i \in V$ at random and fix its value then the expected decrease in the variance of all the other variables is given by,
$$ \E_{i \in V, \{X_i\}} \left[\E_{j \in V} \Var[X_j|X_i] \right] - \E_{j \in V} \Var[X_j]  =   \E_{i,j \in V}  \Cov(X_i,X_j)^2 \cdot \frac{1}{2}\left(\frac{1}{\Var[X_i]} +\frac{1}{\Var[X_j]}  \right) \mper
$$
The above bound is proven in a more general setting in \pref{lem:conditional-variance-correlation}.
As all random variables involved have variance at most $1$, we can rewrite the above expression as,
$$ \E_{i \in V, \{X_i\}} \left[\E_{j \in V} \Var[X_j|X_i] \right] - \E_{j \in V} \Var[X_j]  \geq   \E_{i,j \in V}  |\Cov(X_i,X_j)|^2 \mper
$$
The decrease in the variance is directly related to the {\it global correlations} between random pairs of vertices $i,j \in V$.

Recall that, the failure of independent sampling yields a lower bound on the average local correlations on the edges namely, $E_{i,j \in E} |\Cov(X_i,X_j)|$.
The crucial observation is that if the graph $G$ is a {\it good expander}
in a suitable sense,  then these local correlations translate in to
non-negligible global correlations.  Formally, we show the following (in \pref{sec:local-global-lowrank}):

\begin{lemma}
\torestate{
\label{lem:local-vs-global-low-rank}
  Let $\vec v_1,\ldots,\vec v_n$ be vectors in the unit ball.
  Suppose that the vectors are correlated across the edges of a regular
  $n$-vertex graph $G$,
  \begin{displaymath}
    \E_{ij\sim G} \iprod{\vec v_i,\vec v_j} \ge \rho\mper
  \end{displaymath}
  Then, the global correlation of the vectors is lower bounded by
  \begin{displaymath}
    \E_{i,j\in V} \abs{\iprod{\vec v_i,\vec v_j}}\ge
    \Omega(\rho)/\rank_{\ge \Omega(\rho)}(G)\mper
  \end{displaymath}
where $\rank_{\geq \rho}(G)$ is the number of eigenvalues of adjacency matrix of $G$ that are larger than $\rho$.
}
\end{lemma}

As random variables $X_i$ arise from the solution to a SDP, the matrix $\left(\Cov(X_i,X_j)\right)_{i,j \in V}$ is positive semidefinite, i.e., there exists vectors $u_i$ such that $\iprod{u_i,u_j} = \Cov(X_i,X_j) \,\, \forall i,j \in V$.  Let us consider the vectors $v_i = u_i^{\otimes 2}$.  Suppose the local correlation $\E_{i,j \in E} |\Cov(X_i,X_j)|$ is at least $\epsilon$ then we have,
$$ \E_{i,j \in E} \iprod{v_i,v_j} = \E_{i,j \in E} |\Cov(X_i,X_j)|^2 \geq \eps^2 \mcom $$
and $\E_i [\norm{v_i}^2] \leq 1$.
If the graph $G$ is low-rank, then by \pref{lem:local-vs-global-low-rank} we get a lower bound on the global correlation of the vectors $v_i$, namely
$$ \E_{i,j \in V} |\Cov(X_i,X_j)|^2 = \E_{i,j \in V} \iprod{v_i,v_j} \geq \Omega(\epsilon^2)/\rank_{\geq \eps^2}(G) \mper$$

Summarizing, if the independent sampling is on average $\epsilon$-far from
correlated sampling over the edges, then conditioning on the value of a
random vertex $i \in V$ reduces the average variance by $\eps^2/\rank_{\geq
  \eps^2}(G)$ in expectation.  The same argument can now be applied on the
variables obtained after conditioning on $i$.  In fact, starting with an
SDP solution to $m$-round Lasserre hierarchy, the local distributions
remain consistent and their covariance matrices remain semidefinite as long
as we condition on at most $m-2$ vertices.  Observe that average variance
is at most $1$.  Hence, after at most $\rank_{\geq \eps^2}(G)/\eps^2$
steps, the independent sampling distribution will be within average
distance $\epsilon$ from the correlated sampling on the edges.  The details of this argument are presented in \pref{thm:propagation-sampling}.

\section{General 2-CSP on Low Rank Graphs}
\label{sec:general-2-csp}

Let $\Ins$ be a (general) \maxtwocsp instance with variable set $V=[n]$
and label set $[k]$.
(We represent $\Ins$ as a distribution over triples $(i,j,\Pi)$, where
$i,j\in V$ and $\Pi\sse [k]\times [k]$ is an arbitrary binary predicate.
The goal is to find an assignment $x\in [k]^V$ that maximizes the
probability $\Prob[(i,j,\Pi)\sim \Ins]{(x_i,x_j)\in \Pi}$.)

For simplicity,%
\footnote{If the constraint graph is not regular, all of our results still
  hold for an appropriate definition of threshold rank.}
we will assume that the constraint graph of $\Ins$ is regular, i.e., every
variable $i\in V$ appears in the same number of constraints.
(Since we allow the constraints to be weighted, the precise condition is
that the total weight of the constraints incident to a vertex is the same
for every vertex.)

Let $X_1,\ldots,X_n$ be $r$-local random variables with range $[k]$.
We write $X_\ia$ to denote the $\bits$-indicator of the event $X_i=a$.
Notice that $\set{X_\ia}_{i\in V,\, a\in [k]}$ are also $m$-local random
variables.

\newcommand{\Dist}{\Delta}

For two random variables $X$ and $X'$ with the same range, we denote their
\emph{statistical distance},
\begin{displaymath}
  \normo{\,\set X-\set{X'}\,}
  \defeq \sum_{x} \bigabs{\Prob{X=x}-\Prob{X'=x}}\mper
\end{displaymath}

\paragraph{Independent Sampling and Pairwise Correlation}

The following lemma shows that the statistical difference between
independent sampling and correlated sampling is explained by local
correlation.
\begin{lemma}
  \label{lem:stat-dist-correlation}
  For any two vertices $i,j\in V$,
  \begin{displaymath}
    \Normo{\set{X_i X_j}- \set{X_i}\set{X_j}}
    = \sum_{(a,b)\in [k]^2} \bigabs{\Cov(X_\ia,X_\jb)}
    \mper
  \end{displaymath}
\end{lemma}
\begin{proof}
  Under the distribution $\set{X_iX_j}$, the event $\set{X_i=a, X_j=b}$ has
  probability $\E X_\ia X_\jb$.
  On the other hand, under the product distribution $\set{X_i}\set{X_j}$,
  this event has probability $\E X_\ia  \E X_\jb$.
  Hence, the difference of these probabilities is equal to $\E X_\ia X_\jb
  - \E X_\ia \E X_\jb=\Cov(X_\ia,X_\jb)$.
\end{proof}

\paragraph{Conditional Variance and Pairwise Correlation}
\label{sec:local-corr-cond}

The following lemma shows that conditioning on a variable $X_j$ decreases
the variance of a variable $X_i$ by the correlation of the variables
$X_\ia$ and $X_\jb$.

\begin{lemma}
  \label{lem:conditional-variance-correlation}
  For any two vertices $i,j\in V$,
  \begin{displaymath}
     \Var X_i -\E_{\set{X_j}} \Var\Brac{X_i\Mid X_j}
    \ge  \tfrac 1k \sum_{a,b\in [k]} \E_{\set{X_\ia X_\jb}}
    \Cov(X_\ia,X_{jb})^2 / \Var X_\jb
  \end{displaymath}
\end{lemma}

\begin{proof}
  If we condition on $X_{\jb}$, the variance of $X_\ia$ decreases by
  $\Cov(X_\ia,X_\jb)^2/\Var X_\jb$ (\pref{lem:condition-indicator}).
  Thus, the variance of $X_i$ deceases by $\sum_a \Cov(X_\ia,X_\jb)^2/\Var X_\jb$.
  Hence, there exists $b_0$ such that conditioning on $X_{j b_0}$ causes a
  variance decrement of at least $\tfrac 1k \sum_{a,b}\Cov(X_\ia,X_\jb)^2/\Var X_\jb$.
  Since the variance is non-increasing under conditioning, the variance of
  $X_i$ decreases by at least this amount when we condition on $X_j$.
\end{proof}

\paragraph{Pairwise Correlations and Inner Products}
\label{sec:embedd-corr}

The previous paragraphs were about two different notions of pairwise
correlation.
On the one hand, $\normo{\,\set{X_iX_j}-\set{X_i}\set{X_j}\,}$ and on the
other hand, $\Var X_i-\E_{\set {X_j}} \Var\brac{X_i\mid X_j}$.
The following lemma relates these two notions of pairwise correlations and
shows they can be approximated by inner products of vectors.

\begin{lemma}
\label{lem:general-tensoring-trick}
  Suppose that the matrix $\Paren{\Cov(X_\ia,X_\jb)}_{i\in V,\, a\in[k]}$
  is positive semidefinite.
  Then, there exists vectors $\vec v_1,\ldots,\vec v_n$ in the unit ball
  such that for all vertices $i,j\in V$,
  \begin{displaymath}
    \tfrac1{k^2} \Bigparen{\sum_{(a,b)\in [k]^2} \bigabs{\Cov(X_\ia,X_\jb)}}^2
    \le \iprod{\vec v_i,\vec v_j}
    \le \tfrac1k\sum_{(a,b)\in [k]^2} \tfrac 12(\tfrac{1}{\Var X_\ia} +
    \tfrac{1}{\Var X_\jb})
    \Cov(X_\ia,X_\jb)^2
    \mper
  \end{displaymath}
\end{lemma}

\begin{proof}
  Let $\set{u_\ia }$ be the collection of vectors such that
  $\iprod{u_\ia,u_\jb} = \Cov(X_\ia,X_\jb)$.
  Note that $\snorm{u_\ia}=\Var X_\ia$.
  Define $\vec v_i \seteq k^{-1/2} \sum_{a} u_\ia ^{\tensor 2}/\norm{u_\ia}$.
  (Here, $\bar x$ denote the unit vector in direction~$x$.)
  The inner product of $\vec v_i$ and $\vec v_j$ is equal to
  \begin{displaymath}
    \iprod{\vec v_i,\vec v_j}
    =\tfrac1k\sum_{a,b} \tfrac1 {\sqrt{\Var X_\ia\Var X_\jb}} \Cov(X_\ia,X_\jb)^2\mper
  \end{displaymath}
  Using the inequality between arithmetic mean and geometric mean, we have
  $\paren{\Var X_\ia\Var X_\jb}^{-1/2}\le \paren{1/\Var X_\ia + 1/\Var
    X_\jb}/2$, which implies the desired upper bound on the inner product
  $\iprod{\vec v_i,\vec v_j}$.
  On the other hand,  by Cauchy--Schwartz,
  \begin{displaymath}
    \Bigparen{\sum_{a,b} \bigabs{\Cov(X_\ia,X_\jb)}}^2
    \le \sum_{a,b}  {\sqrt{\Var X_\ia\Var X_\jb}} \cdot \sum_{a,b} \tfrac1
    {\sqrt{\Var X_\ia\Var X_\jb}} \Cov(X_\ia,X_\jb)
    \mper
  \end{displaymath}
  Since $\sum_a \Var X_\ia \le \sum_a \E X_\ia^2 =1$ , we have $\sum_a
  \sqrt{\Var X_\ia} \le \sqrt k$ for all vertices $i\in V$ (by
  Cauchy--Schwartz).
  Therefore,
  \begin{displaymath}
    \Bigparen{\sum_{a,b} \bigabs{\Cov(X_\ia,X_\jb)}}^2
    \le k \sum_{a,b} \tfrac1
    {\sqrt{\Var X_\ia\Var X_\jb}} \Cov(X_\ia,X_\jb)\mcom
  \end{displaymath}
  which gives the desired lower bound on the inner product $\iprod{\vec
    v_i,\vec v_j}$.
  It remains to argue that the vectors $\vec v_1,\ldots,\vec v_n$ are
  contained in the unit ball.
  Since $\Cov(X_\ia,X_{ib})^2\le \Var X_\ia \Var X_{ib}$, we can upper
  bound $\snorm{\vec v_i}\le k^{-1} \sum_{a,b} \sqrt{\Var X_\ia \Var
    X_{ib}}\le 1$ (using $\sum_a \sqrt{\Var X_\ia}\le \sqrt k$).
\end{proof}

\paragraph{Local Correlation vs Global Correlation on Low-Rank Graphs}
\label{sec:local-vs-global}

The following lemma shows that local correlation (correlation across
edges of a graph) implies global correlation (correlation
between random vertices) if the graph has low threshold rank.
(Proof in \pref{sec:local-global-lowrank}.)

\restatelemma{lem:local-vs-global-low-rank}

\paragraph{Putting Things Together}
\label{sec:putt-things-togeth}

The following lemma shows that either independent sampling is
statistically close to correlated sampling across edges of a graph
or the typical variance of a vertex decreases non-trivially by conditioning
on a random vertex.

\begin{lemma}
  \label{lem:one-step-conditioning}
  Let $G$ be a regular $n$-vertex graph and $\e$ be the expected
  statistical distance between independent and correlated sampling across
  the edges of $G$,
  \begin{displaymath}
    \e = \E_{ij\sim G}~
    \Normo{\set{X_i X_j}- \set{X_i}\set{X_j}}
  \end{displaymath}
  Further, suppose that the matrix $\Paren{\Cov(X_\ia,X_\jb)}_{i\in V,\, a\in[k]}$
  is positive semidefinite.
  Then, conditioning on a random vertex decreases the variances by
  \begin{displaymath}
    \E_{i,j\in V} \E_{\set{X_j}} \Var\Brac{X_i\Mid X_j}
    \le \E_{i\in V} \Var X_i - \Omega(\e^2/k)/\rank_{\ge \Omega(\e/k)^2}(G)\mper
  \end{displaymath}
\end{lemma}

\begin{proof}
  Let $\vec v_1,\ldots,\vec v_n$ be the vectors constructed in
  \pref{lem:general-tensoring-trick}.
  By \pref{lem:general-tensoring-trick} and
  \pref{lem:stat-dist-correlation}, the local correlation of these vectors
  is at least
  \begin{displaymath}
    \E_{ij\sim G}\iprod{\vec v_i,\vec v_j}
    \ge \tfrac 1{k^2} \E_{ij\sim G} \Normo{\set{X_i X_j}-
      \set{X_i}\set{X_j}}^2
    \ge \e^2/k^2\mper
  \end{displaymath}
  (The last step also uses Cauchy--Schwartz.)
  Hence, \pref{lem:local-vs-global-low-rank} implies the following lower
  bound on the global correlation of these vectors,
  \begin{displaymath}
    \E_{i,j\in V} \abs{\iprod{\vec v_i,\vec v_j}}\ge \Omega(\e/k)^2
    /\rank_{\ge \Omega(\e/k)^2}(G)\mper
  \end{displaymath}
  \pref{lem:general-tensoring-trick} and
  \pref{lem:conditional-variance-correlation} allows us to relate the
  expected decrement of the variances to the global correlation of the
  vectors $\vec v_1,\ldots, \vec v_n$,
  \begin{displaymath}
    \E_{i,j\in V} \Bigbrac{\Var X_i - \E_{\set{X_j}}\Var\brac{X_i\mid X_j} }
    \ge k \cdot \E_{i,j\in V} \abs{\iprod{\vec v_i,\vec v_j}}\mcom
  \end{displaymath}
  which gives the desired upper bound on $\E_{i,j\in V}\E_{\set{X_j}}\Var\brac{X_i\mid X_j}$.
\end{proof}

The following lemma asserts that if the constraint graph has low threshold
rank then there exists a partial assignment $x_S$ to a small set $S$ of
vertices such that independent sampling conditioned on this assignment
$x_S$ gives almost the same value as correlated sampling (without
conditioning on the assignment $x_S$).

  \begin{mybox}
\begin{algorithm}[Propagation Sampling]\mbox{}\label{alg:propagation-sampling}
  \begin{description}
  \item[Input:] $r$-local random variables $X_1,\ldots,X_n$ over $[k]$
  \item[Output:] (global) distribution over assignments $x\in [k]^V$.
  \end{description}
  \begin{enumerate}
  \item Choose $m \in \set{1,\ldots,r}$ at random.
  \item Sample a random set of ``seed vertices'' $S\in V^m$.
    (Repeated vertices are allowed.)
  \item Sample a assignment $x_S\in [k]^S$ for $S$ according to its local
    distribution $\set{X_S}$.

  \item For every other vertex $i\in V\sm S$, sample a label $x_i\in[k]$
    according to the local distribution for $S\cup \set{i}$ conditioned on
    the assignment $x_S$ for $S$.
  \end{enumerate}
\end{algorithm}
\end{mybox}

\begin{theorem}
  \label{thm:propagation-sampling}
  \label{thm:general-local-to-global}
  Let $X_1,\ldots,X_n$ be $r$-local random variables and let
  $X'_1,\ldots,X'_n$ be the random variables produced by
  \pref{alg:propagation-sampling} on input $X_1,\ldots,X_n$.
  Suppose that the matrices $\paren{\Cov(X_\ia,X_\jb\mid X_S=x_S)}_{i\in V,\, a\in[k]}$
  are positive semidefinite for every set $S\sse V$ with $\card S \le r$
  and local assignment $x_S\in [k]^S$.
  Then, if $r \gg   O(k/\e^4)\cdot \rank_{\Omega(\e/k)^2}(G)$,
  \begin{displaymath}
    \E_{ij\sim G} \Normo{\set{X_i X_j}-\set{X'_iX'_j}}
    \le \e
    \mper
  \end{displaymath}
\end{theorem}

\begin{proof}
  Let us define $\e_m$,
  \begin{displaymath}
    \e_m = \E_{S\in V^m} \E_{\set{X_S}} \E_{ij\sim G} \Normo{\,\set{X_i X_j \mid X_S} -
      \set{X_i\mid X_S}\set{X_j\mid X_S}\,}\mper
  \end{displaymath}
  To prove the current theorem it is enough to show that $\E_{m\in[r]} \e_m\le \e$.
  For $m\le r$, define a non-negative potential $\Phi_m$ as follows
  \begin{displaymath}
    \Phi_m
    \seteq \E_{S\in V^m} \E_{\set{X_S}} \E_{i\in V} \Var(X_i\mid X_S)
    \mper
  \end{displaymath}
  Let $m\in [r]$. Suppose $\e_m\ge \e/2$.
  Then,
  \begin{displaymath}
    \Prob[S\in V^m, \set{X_S}]{\E_{ij\sim G} \Normo{\,\set{X_i X_j \mid X_S} -
      \set{X_i\mid X_S}\set{X_j\mid X_S}\,}\ge\e/2 }\ge\e/2\mper
  \end{displaymath}
  Therefore, by \pref{lem:one-step-conditioning},
  \begin{displaymath}
    \E_{S\in V^m, \set{X_S}} \Bigbrac{\E_{i\in V} \Var\brac{X_i\mid X_S} -
      \E_{i,j\in V} \Var\brac{X_i\mid X_S,X_j} } \ge \e/2 \cdot
    \Omega(\e^2/k) / \rank_{\ge \Omega(\e/k)^2}(G)\mper
  \end{displaymath}
  In other words,
  \begin{math}
    \Phi_{m+1} \le \Phi_{m} - \Omega(\e^3/k)/\rank_{\ge \Omega(\e/k)^2}(G).
  \end{math}
  Since $1\ge \Phi_1\ge \ldots\ge \Phi_r\ge 0$, it follows that there are
  at most
  \begin{math}
    O(k/\e^3) \cdot \rank_{\ge \Omega(\e/k)^2}(G)
  \end{math}
  indices $m\in [r]$ such that $\e_m\ge \e/2$.
  Therefore, if $r\gg O(k/\e^4) \cdot \rank_{\ge \Omega(\e/k)^2}(G)$, we have
  \begin{displaymath}
    \E_{m\in[r]} \e_m \le \e/2 + \tfrac{1}{r} \cdot O(k/\e^3) \cdot
    \rank_{\ge \Omega(\e/k)^2}(G)\le \e\mper
  \end{displaymath}
  Finally, by the triangle inequality,
  \begin{align*}
    &\E_{ij\sim G} \Normo{\set{X_i X_j}-\set{X'_iX'_j}}
    \\
    &= \E_{ij\sim G} \Bignormo{\, \Bigparen{\E_{m\in[r]}\E_{S\in
          V^m}\E_{\set{X_S}} \set{X_iX_j\mid X_S} }-
      \Bigparen{\E_{m\in[r]}\E_{S\in V^m}\E_{\set{X_S}} \set{X_i\mid
          X_S}\set{X_j\mid X_S} } \, }
    \\
    & \le \E_{ij\sim G} \E_{m\in[r]}\E_{S\in V^m}\E_{\set{X_S}}
    \Normo{\set{X_iX_j\mid X_S} - \set{X_i\mid X_S}\set{X_j\mid X_S} }
    = \E_{m\in [r]} \e_m \le \e\mper\qedhere
  \end{align*}
\end{proof}

The following theorem directly implies \pref{thm:main}.

\begin{theorem}
  Let $\e>0$ and $r=O(k)\cdot \rank_{\ge \Omega(\e/k)^2}(G)/\e^4$.
  Suppose that the $r$-round Lasserre value of the \maxtwocsp instance
  $\Ins$ is $\sigma$.
  Then, given an optimal $r$-round Lasserre solution,
  \pref{alg:propagation-sampling} (Propagation Sampling) outputs an
  assignment with expected value at least $\sigma-\e$ for $\Ins$ .
\end{theorem}

\begin{proof}
  An  optimal $r$-round Lasserre solution gives rise to $r$-local random
  variables $X_1,\ldots,X_n$ over $[k]$.
  Let $X_\ia$ be the indicator variable of the event $X_i=a$.
  The matrices $\set{\Cov(X_\ia,X_\jb\mid X_S=x_S)}_{i,j\in V,\,a,b\in
    [k]}$ are positive semidefinite for all sets $S\sse V$ with $\card S
  \le r$ and local assignments $x_S\in [k]^S$.
  Furthermore, the Lasserre solution satisfies
  \begin{displaymath}
    \E_{(i,j,\Pi)\sim \Ins)} \Prob[\set{X_iX_j}]{(X_i,X_j)\in \Pi}=\sigma\mper
  \end{displaymath}
  Let $X_1',\ldots,X'_n$ be the jointly-distributed (global) random
  variables in \pref{thm:general-local-to-global}.
  By \pref{thm:general-local-to-global}, we can estimate the expected value
  of the assignment $X'_1,\ldots,X'_n$ as
  \begin{align*}
    \E_{X'_1,\ldots,X'_n} \Prob[(i,j,\Pi)\sim \Ins]{(X'_i,X'_j)\in \Pi}
    &= \E_{(i,j,\Pi)\sim \Ins} \Prob[\set{X'_iX'_j}]{(X'_i,X'_j)\in \Pi}
    \\
    &\ge \E_{(i,j,\Pi)\sim \Ins} \Prob[\set{X_iX_j}]{(X_i,X_j)\in \Pi}
    - \tfrac 12 \E_{ij\sim G} \Normo{\, \set{X_iX_j} - \set{X'_iX'_j}\,}
    \\
    &\ge \sigma-\e\mper\qedhere
  \end{align*}

\end{proof}

\subsection{Special case of Unique Games}

The following lemma is a version of \pref{lem:general-tensoring-trick}
tailored towards \uniquegames.
The advantage of this version of the lemma is that the bounds are
independent of the alphabet size~$k$.

\begin{lemma}
  \torestate{
  \label{lem:ug-tensoring-trick}
  Let $X_1,\ldots,X_n$ be $r$-local random variables over $[k]$ and let
  $X_\ia$ be the indicator of the event $X_i=a$.
  Suppose that the matrix $\Paren{\Cov(X_\ia,X_\jb)}_{i\in V,\, a\in[k]}$
  is positive semidefinite.
  Then, there exists vectors $\vec v_1,\ldots,\vec v_n$ in the unit ball
  such that for all vertices $i,j\in V$ and permutations $\pi$ of $[k]$,
  \begin{displaymath}
    \Bigparen{\sum_{a\in [k]} \bigabs{\Cov(X_\ia,X_{j\,\pi(a)})}}^4
    \le \iprod{\vec v_i,\vec v_j}
    \le \sum_{(a,b)\in [k]^2} \tfrac{1}{2}(\tfrac1{\Var X_\ia}+\tfrac1{ \Var X_\jb})
    \Cov(X_\ia,X_\jb)^2\mper
  \end{displaymath}
}
\end{lemma}

The following theorem immediately implies \pref{thm:main-ug}.
Let $\Ins$ be a \uniquegames instance with alphabet size $k$ and constraint
graph $G$.

\begin{theorem}
  Let $\e>0$ and $r=k\cdot \rank_{\ge \Omega(\e^4)}(G)/\e^{O(1)}$.
  Suppose that the $r$-round Lasserre value of the \uniquegames instance
  $\Ins$ is $\sigma$.
  Then, given an optimal $r$-round Lasserre solution,
  \pref{alg:propagation-sampling} (Propagation Sampling) outputs an
  assignment with expected value at least $\sigma-\e$ for $\Ins$ .
\end{theorem}

\begin{proof}[Proof Sketch]
  Let $X_1,\ldots,X_n$ be $r$-local random variables over $[k]$ from an
  optimal $r$-round Lasserre solution for $\Ins$.
  The local variables satisfy
  \begin{displaymath}
    \E_{(i,j,\pi)\sim \Ins)} \Prob[\set{X_iX_j}]{X_j=\pi(X_i)}=\sigma\mper
  \end{displaymath}
  For a permutation $\pi$ of $[k]$, we define a modified version of
  statistical distance,
  \begin{displaymath}
    \norm{\, \set{X_iX_j} - \set{X_i}\set{X_j}\,}_\pi
    \defeq \sum_a \abs{\Prob[\set{X_iX_j}]{X_j=\pi(X_i)}
      -\Prob[\set{X_i}\set{X_j}]{X_j=\pi(X_i)}}
    \mper
  \end{displaymath}
  The following analog of \pref{lem:stat-dist-correlation} holds,
  \begin{displaymath}
    \norm{\, \set{X_iX_j} - \set{X_i}\set{X_j}\,}_\pi
    =\sum_{a} \abs{\Cov(X_\ia,X_{j\pi(a)})}
    \mper
  \end{displaymath}
  Using \pref{lem:ug-tensoring-trick}, it is straight-forward to prove a
  better versions of \pref{lem:one-step-conditioning} and
  \pref{thm:general-local-to-global} for our modified notion of statistical
  distance.
  The conclusion is that for $r\ge k\cdot \rank_{\ge
    \Omega(\e^4)}(G)/\e^4$, \pref{alg:propagation-sampling} (Propagation
  Sampling) produces (global) random variables $X'_1,\ldots,X'_n$ such that
  \begin{displaymath}
    \E_{(i,j,\pi\sim \Ins)} \norm{\set{X_iX_j} - \set{X'_i-X'_j}}_\pi \le \e\mper
  \end{displaymath}
  Therefore, we can estimate the expected fraction of satisfied constraints
  as
  \begin{align*}
    \E_{X'_1,\ldots,X'_n} \Prob[(i,j,\pi)\sim \Ins]{X'_j=\pi(X'_i)}
        &\ge     \E_{X'_1,\ldots,X'_n} \Prob[(i,j,\pi)\sim \Ins]{X_j=\pi(X_i)}
    -     \E_{(i,j,\pi\sim \Ins)} \norm{\set{X_iX_j} - \set{X'_i-X'_j}}_\pi
    \\
    &\ge \sigma -\e\mper\qedhere
  \end{align*}
\end{proof}

\section{Local Correlation implies Global Correlation in Low-Rank Graphs}
\label{sec:local-global-lowrank}

Let $G$ be a regular graph with vertex set $V=\set{1,\ldots,n}$.
We identify $G$ with its normalized adjacency matrix, a symmetric
stochastic matrix.
Let
\begin{math}
  \lambda_1 \ge \ldots \ge \lambda_n \in [-1,1]
\end{math}
be the eigenvalues of $G$ in non-increasing order.

The following lemma shows that a violation of the local vs global
correlation condition implies that the graph has high threshold rank.


\begin{lemma} \label{lem:sdptothresholdrank}
  Suppose there exist vectors $v_1,\ldots,v_n\in \R^n$ such that
  \begin{gather*}
    \E_{ij\sim G} \iprod{v_i,v_j}\ge 1-\e\mcom
    \quad
    \E_{i,j\in V}\iprod{v_i,v_j}^2\le \tfrac1m\mcom
    \quad
    \E_{i\in V}\snorm{v_i}=1\mper
  \end{gather*}
  Then for all $C > 1$, $\lambda_{(1-1/C)m}\ge 1-C\cdot \e$. In particular, $\lambda_{m/2}>1-2\e$.
\end{lemma}

\begin{proof}
  Let $X=(x_{r,s})_{r,s\in[n]}$ be the Gram matrix $(\iprod{\vec v_i, \vec v_j})_{i,j\in V}$
  represented in the eigenbasis of~$G$, so that
  \begin{gather*}
    \E_{ij\sim G }\iprod{v_i,v_j} = \sum_{r\in [n]} \lambda_rx_{r,r}\mcom
    \quad
    \E_{i,j\in V}\iprod{v_i,v_j}^2 = \sum_{r,s\in [n]} x_{r,s}^2\mcom
    \quad
    \E_{i\in V}\snorm{v_i} = \sum_{r\in [n]} x_{r,r}\mper
  \end{gather*}
  Let $m'$ be the largest index such that $\lambda_{m'}\ge 1-C\cdot \e$.
  Notice that the numbers $p_1=x_{1,1},\ldots,p_n=x_{n,n}$ form a probability
  distribution over $r\in[n]$.
  Let $q=\sum_{i=1}^{m'} p_i$ be the probability of the event $r\le m'$.
  Using Cauchy--Schwarz, we can bound this probability in terms of $m$,
  \begin{displaymath}
    q=\sum_{r=1}^{m'} p_r \le m' \sum_{r=1}^n p_r^2 \le \tfrac {m'}m\mper
  \end{displaymath}
  On the other hand, we can bound the expectation of $\lambda_r$ with
  respect to the probability distribution $(p_1,\ldots,_n)$ in terms of
  this probability $q$,
  \begin{displaymath}
    1-\e
    \le \sum_{r=1}^n \lambda_r p_r
    \le \sum_{r=1}^{m'}p_r + (1-C\cdot \e)\sum_{r=m'+1}^m p_r
    = 1- (1-q) C \cdot \e
    \le 1 -  \Paren{1-\tfrac{ m'}{m}} C\cdot \e
    \mper
  \end{displaymath}
  It follows that $m'\ge \Paren{1-\nfrac 1 C}\cdot m$, which gives the
  desired conclusion that $G$ has at least $\Paren{1-\nfrac 1 C}\cdot m$
  eigenvalues $\lambda_r\ge -C\cdot \e$.
\end{proof}

Note that \pref{lem:local-vs-global-low-rank} follows directly from the previous lemma by picking $C = \frac{(1-\rho/100)}{(1-\rho)}$ and observing that $\E_{i,j \in V} |\iprod{\vec v_i, \vec v_j}| \geq \E_{i,j \in V} |\iprod{\vec v_i, \vec v_j}|^2$ since $|\iprod{\vec v_i, \vec v_j}| \leq 1$ for all $i,j \in V$

As a converse to \pref{lem:sdptothresholdrank}, the following lemma shows that if a graph has many eigenvalues close to
$1$, then there exist vectors for the vertices of the graph with high local
correlation and low global correlation.

\begin{lemma}
  If $\lambda_m\ge1-\e$, then there exist vectors $v_1,\ldots,v_n\in  \R^m$
  such that
  \begin{gather*}
    \E_{ij\sim G} \iprod{v_i,v_j}\ge 1-\e\mcom
    \\
    \E_{i,j\in V}\iprod{v_i,v_j}^2= \tfrac1m\mcom
    \\
    \E_{i\in V}\snorm{v_i}=1\mper
  \end{gather*}
\end{lemma}

\begin{proof}
  Let $\super f1,\ldots,\super f m\from V\to \R$ be orthonormal
  eigenfunctions of $G$ with eigenvalue larger than $1-\e$.
  Consider vectors $v_1,\ldots,v_n\in \R^m$ satisfying
  \begin{math}
    \iprod{v_i,v_j}
    =\E_{r\in [m]}\super f r _i \super f r _j.
  \end{math}
  Since the functions $\super f r$ have norm $1$, the typical squared norm
  of the vectors $v_i$  satisfies
  \begin{displaymath}
    \E_{i\in V} \snorm{v_i}
    = \E_{r\in[m]} \snorm{\super f r}
    =1
    \mper
  \end{displaymath}
  Since the eigenvalues of the eigenfunctions $\super f r$ are larger than
  $1-\e$, we can lower bound the local correlation of the vectors $v_i$,
  \begin{displaymath}
    \E_{ij\sim G} \iprod{v_i,v_j}
    = \E_{r\in [m]}\iprod{\super f r,G \super f r}
    \ge 1-\e
    \mper
  \end{displaymath}
  Finally, since the function $\super f m$ are orthonormal, the global
  correlation of the vectors $v_i$ is
  \begin{displaymath}
    \E_{i,j\in V}\iprod{v_i,v_j}^2
    = \E_{i,j\in V} \E_{r,s\in [m]}
    \super  f r _i \super f r_j \super f  {s}_i \super f s _j
    =\E_{r,s\in [m]} \iprod{\super f r,\super f s}^2
    = \tfrac 1m
    \mper
    \qedhere
  \end{displaymath}
\end{proof}

\begin{remark}
  The condition that there exist vectors $v_1,\ldots,v_n\in \R^n$ with
  \begin{gather*}
    \E_{ij\sim G} \iprod{v_i,v_j}\ge 1-\e\mcom
    \\
    \E_{i,j\in V}\iprod{v_i,v_j}^2\le \tfrac1m\mcom
    \\
    \E_{i\in V}\snorm{v_i}=1\mper
  \end{gather*}
  is equivalent to the condition that there exists a symmetric positive
  semidefinite matrix $X\in \R^{V\times V}$ such that
  \begin{gather*}
    \Tr G X \ge 1-\e \mcom
    \\
    \Tr X^2 \le 1/m\mcom
    \\
    \Tr X =1\mper
  \end{gather*}
\end{remark}

\section{On Low Rank Approximations to Sets of Vectors}
\label{sec:low-rank-approx}

\begin{theorem}
\label{thm:low-rank-approx}
  Let  $v_1,\ldots,v_n\in \R^n$ be vectors in the unit ball.
  Then for every $\e>0$, there exists a subset $U\sse \set{v_1,\ldots,v_n}$
  with $\card{U}\le 1/\e$ such that
  $\E_{i,j\in [n]}\norm{w_i}\, \norm{w_j}\, \iprod{\bar w_i,\bar w_j}^2\le \e$,
  where $w_i$
  is the projection of $v_i$ to the orthogonal complement of the span of
  $U$.
\end{theorem}

The proof of \pref{thm:low-rank-approx} is by an iterative construction.
In each iteration, we will use the following lemma.

\begin{lemma}
\label{lem:approx-step}
  Let $v_1,\ldots,v_n\in \R^n$ be vectors.
  %
  Then, there exists a unit vector $u\in \set{\bar v_1,\ldots,\bar v_n}$
  such that the vectors $v'_1,\ldots,v'_n$ with $v'_i = v_i-\iprod{v_i,u}u$
  satisfy the following condition,
  \begin{displaymath}
    \E_{i\in [n]} \snorm{v'_i}
    \le \E_{i\in [n]}\snorm{v_i}- \E_{i,j\in [n]}
    \norm{v_i}\,\norm{v_j}\,\iprod{\bar v_i,\bar v_j}^2\mper
  \end{displaymath}
\end{lemma}

\begin{proof}
  Suppose we pick a random index $j\in [n]$ and choose $u=\bar v_j$.
  In this case, the squared norm of the vectors $v'_i=v_i-\iprod{v_i,u}u$
  equals
  \begin{displaymath}
    \snorm{v'_i}=\snorm{v_i}-\iprod{v_i,u}^2=(1-\iprod{\bar
      v_i,\bar v_j}^2)\snorm{v_i}\mper
  \end{displaymath}
  Hence, we can estimate the expected decrease of the typical squared norms
  for a random vector $u\in \set{\bar v_1,\ldots,\bar v_n}$.
  \begin{align*}
    \E_{i\in[n]}\snorm{v_i}-\E_u \E_{i\in[n]} \snorm{v'_i}
    & = 
    \E_{i,j\in[n]} \Paren{ %
      \tfrac12 \snorm{v_i}+\tfrac12\snorm{v_j}} \iprod{\bar v_i,\bar v_j}^2
    \\
    &\ge \E_{i,j\in[n]}\norm{v_i}\,\norm{v_j}\,\iprod{\bar v_i,\bar v_j}^2
  \end{align*}
  It follows that there exists a unit vector
  \begin{math}
    u\in \set{\bar v_1,\ldots,\bar v_n}
  \end{math}
  such that the vectors $v'_i=v_i-\iprod{v_i,u}u$ have the desired property
  \begin{displaymath}
    \E_{i\in [n]}\snorm{v_i'}
    \le \E_{i\in[n]}\snorm{v_i}
    -\E_{i,j\in[n]}\norm{v_i}\,\norm{v_j}\,\iprod{\bar   v_i,\bar v_j}^2
    \mper
    \qedhere
  \end{displaymath}

\end{proof}

\begin{proof}[Proof of \pref{thm:low-rank-approx}]
  We can construct the set $U$ in a greedy fashion so as to minimize the total
  squared norm of the vectors $w_1,\ldots,w_n$ (the projections of the
  vectors $v_i$ to the orthogonal complement of the span of $U$).
  (In fact, we could choose set $U$ randomly.)
  To make the analysis more convenient, we use the following, slightly
  different construction.
  \begin{enumerate}
  \item Let $\super v 1_i=v_i$ for all $i\in [n]$.
  \item For $t$ from $1$ to $1/\e$, construct vectors $\super u t\in \R^n$ and $\super
    v {t+1}_1,\ldots, \super v {t+1}_n\in \R^n$ as follows:
    \begin{enumerate}\item
      Using \pref{lem:approx-step}, pick a unit vector
      \begin{math}
        \super u t\in \set{\super {\bar v} {t} _1,\ldots,\super {\bar
            v}{t}_n}
      \end{math}
      such that the vectors $\super v {t+1}_i= \super v {t}_i-\iprod{\super v
        {t}_i,\super u t}\super u t$  satisfy the condition
      \begin{displaymath}
        \E_{i\in[n]}\snorm{\super v {t+1}_i}
        \le \E_{i\in[n]}\snorm{\super v{t}_i}
        -\E_{i,j\in[n]}\norm{\super v {t}_i}\,\norm{\super v {t}_j}
        \, \iprod{\super v {t}_i,\super v {t}_j}^2\mper
      \end{displaymath}
    \end{enumerate}
  \end{enumerate}
  Notice that the vectors $\super v t _1,\ldots,\super v t_n$ are
  the projections of the vector $v_1,\ldots, v_n$ into the orthogonal
  complement of the span of the vectors $\super u 1,\ldots,\super u {t-1}$.
  Let $U$ be the set of all indices $j$ such that $\super u t=\super {\bar v} t_j$ for
  some $t\in\set{1,\ldots,1/\e}$.
  We can verify that the vectors $\super u 1,\ldots,\super u {1/\e}$ are an
  orthonormal basis of the span of $U$.
  Let $w_1,\ldots,w_n$ be the projections of the vectors $v_1,\ldots,v_n$
  into the orthogonal complement of the span of $U$ (so that $w_i=\super v
  {1/\e}_i$).
  Since the vectors $w_1,\ldots,w_n$ are projections of the vectors $\super
  v t_1,\ldots,\super v t_n$ for all $t\in {1,\ldots,1/\e}$, it follows
  that
  \begin{displaymath}
    \E_{i,j\in[n]}\norm{w_i}\,\norm{w_j}\,\iprod{\bar w_i,\bar w_j}^2
    \le \E_{i,j\in[n]}\norm{\super v {t}_i}\,\norm{\super v {t}_j}
    \, \iprod{\super v {t}_i,\super v {t}_j}^2\mper
  \end{displaymath}
  Hence, we can bound the typical squared norm of the vectors $w_i$,
  \begin{displaymath}
    \E_{i\in [n]}\snorm{w_i}\le \E_{i\in[n]}\snorm{v_i}- \tfrac 1 \e
    \E_{i,j\in[n]}\norm{w_i} \,\norm{w_j}\, \iprod{\bar w_i,\bar w_j}^2\mper
  \end{displaymath}
  Since the left-hand side is nonnegative and $\E_{i\in [n]}\snorm{v_i}\le
  1$, it follows that
  \begin{math}
    \E_{i,j\in[n]}\norm{w_i} \,\norm{w_j}\, \iprod{\bar w_i,\bar w_j}^2\le \e\mcom
  \end{math}
  as desired.
\end{proof}

For our applications it will sometimes be convenient to associate different
subspace with subsets $U$ of vectors (in \pref{thm:low-rank-approx}, we
associate the span of vectors in $U$ with the subset~$U$).

\begin{theorem} \label{thm:low-rank-approx-general}
  Let  $v_1,\ldots,v_n\in \R^n$ be vectors in the unit ball.
  For every subset $U\sse V$, let $Q_U$ be the projector on some subspace
  orthogonal to the span of $U$.
  (Note that $Q_U$ is not necessarily the projector on the orthogonal
  complement of the span of $U$.)
  Then for every $\e>0$, there exists a subset $U\sse \set{v_1,\ldots,v_n}$
  with $\card{U}\le 1/\e$ such that
  $\E_{i,j\in [n]}\norm{w_i}\, \norm{w_j}\, \iprod{\bar w_i,\bar w_j}^2\le \e$,
  where $w_i=Q_U v_i$.
\end{theorem}

\begin{proof}
  We use the same construction as in the proof of
  \pref{thm:low-rank-approx}.
  The only difference is that we define $\super v {t+1}_i=P_{\super U t}
  v_i$ (instead of $\super v {t+1}_i= \super v {t}_i-\iprod{\super v
    {t}_i,\super u t}\super u t$).
  Here, $\super U t$ is the set of all indices $j$ such that $\super u {t'} = \super {\bar v_j}
  {t'}$ for some $t'\le t$.
  The proof is still applies to this modifies construction because
  $\norm{\super v {t+1}_i}\le \norm{\super v {t}_i-\iprod{\super v
      {t}_i,\super u t}\super u t}$ (which is the only fact used about
  these vectors).
\end{proof}

\section{Rounding SDP Solutions to Unique Games}
\label{sec:round-uniq-games}

	In this section, we will present a subexponential time algorithm for Unique Games based on a SDP hierarchy, namely the simple SDP augmented with Sherali-Adams hierarchy.  This hierarchy of relaxations weaker than the Lasserre hierarchy was studied in some earlier works \cite{RaghavendraS09c,KhotS09}.   Roughly speaking, the $m$th round relaxation in this hierarchy corresponds to the basic semidefinite program, along with all valid constraints on at most $m$ vectors.  Formally, the variables in the $m$th round relaxation for Unique Games consists of

\begin{itemize}
\item A collection of ``local distributions''
    \begin{math}
      \set{\mu_T}_{T\sse V, ~\card T \le m}.
    \end{math}
    Each distribution $\mu_T$ is over local assignments $x_T\in [k]^T$.
\item A set of vectors $\cV = \set{v_{ia}}{ i \in V, a \in [k]}$ with $k$ orthogonal vectors for every vertex $i \in V$.
\end{itemize}
The constraint of the SDP relaxation ensure that the inner products of the vectors are consistent with the corresponding local distributions, i.e.,
for all $S \sse V$, $\card{S} \leq m$  $i,j \in S$ and $a,b \in [k]$,
$$ \Pr_{\mu_S}\left[X_i = a \wedge X_j = b\right] = \iprod{v_{ia},v_{jb}} \mper$$
The objective value of the SDP corresponds to minimizing the number of violated constraints,
$$\mathrm{Minimize} \ \ \ \ \ E_{i,j \in E} \left[\sum_{a \in [k]} \norm{v_{ia}-v_{j\pi_{ij}(a)}}^2\right] \mper$$

	\subsection{Propagation Rounding} \label{sec:propagation}

	Let  $\set{\mu_T}_{T\sse V, ~\card T \le m}$ be a set of consistent local distributions over assignments.  For a subset of vertices $S$, the distribution $\mu^{|S}$ over global assignments is sampled as follows:
  \begin{enumerate}
  \item Sample a assignment $x_S\in [k]^S$ for $S$ according to its local
    distribution $\mu_S$.

  \item For every other vertex $i\in V\sm S$, sample a label $x_i\in[k]$
    according to the local distribution for $S\cup \set{i}$ conditioned on
    the assignment $x_S$ for $S$.
  \end{enumerate}
The above procedure will be referred to as {\it propagation rounding} and the set $S$ of vertices will be called the {\it seed} vertices.

The following lemma implies that if the seed vertices $S$ nearly
determine the values of a set of vertices $T$, then the assignment
output by the propagation rounding has a distribution similar to the
local distribution $\mu_T$ that is part of the LP/SDP solution (hence gets close to the SDP value).

\begin{lemma} \label{lem:propagationroundingerror} For a set $S \sse V$ $|S| = m-t$, let $\mu^{|S}$ denote the distribution over global
	assignments $x \in [k]^V$ output by propagation rounding with
	$S$ as the seed vertex set.  Then, for every subset $T$ with
	$|T| \leq t$ we have
	\beq    \Norm{\,\mu_{T} - \mu^{|S}_{T}\,} \leq \sum_{t \in T}
	\Var[X_t|X_S] \mper \eeq
\end{lemma}

\begin{proof}
Consider the following experiment,
\begin{enumerate}
\item Sample a assignment $x_S\in [k]^S$ for $S$ according to its local
    distribution $\mu_S$,
\beq
      x_S\sim  \mu_S\mper
    \eeq

   \item Sample an assignment $y_T \in [k]^T$ according to the local
	   distribution for $S \cup T$ conditioned on the assignment
	   $x_S$ for $S$,
	    \beq
      y_T \sim \mu_{S\cup T } \mid x_S\mper
    \eeq

  \item For every vertex $t\in T$, sample a label $x_t\in[k]$
    according to the local distribution for $S\cup \set{t}$ conditioned on
    the assignment $x_S$ for $S$,
\beq
      x_t \sim \mu_{S,t } \mid x_S\mper
    \eeq
\end{enumerate}
Clearly the distribution of $y_T$ is $\mu_T$, while the distribution of
$x_{T}$ is $\mu^{|S}_{T}$.  For any $t \in T$, the coordinates $x_t$
and $y_t$ are independent samples from $\mu_{S,t } \mid x_S$.
Therefore we have,
\beq \Pr[x_t \neq y_t | x_S] = 1 - \cp(\{x_t | x_S\}) = \Var\left[(X_t
|x_S) \right] \mper \eeq
By a union bound we get,
\beq \Pr[x_T \neq y_T | x_S] =  \sum_{t \in T} \Var\left[(X_t
|x_S) \right] \mper \eeq
Averaging over the different choices of $x_S$,
\beq \Pr[x_T \neq y_T] = \E_{x_S} \left[\sum_{t \in T} \Var\left[(X_t
|x_S) \right] \right] = \sum_{t \in T} \Var\left[X_t
|X_S \right] \mper \eeq
Therefore, the statistical distance between the distributions $\mu_T$
and $\mu^{|S}_{T}$ associated with $y_T$ and $x_T$ is atmost $\sum_{t \in T} \Var\left[X_t
|X_S \right]$.
\end{proof}

\subsection{Unique Games on Low Rank Graphs}

	Let $G$ be an instance of unique games whose constraint graph
	$G$ has low threshold rank.  Let $\mathcal{V} =
	\set{v_{ia}}_{i \in V, a \leq [k]}$ be an SDP solution for
	$G$, and let $\{ \mu_S \}_{S \sse V, |S| \leq m}$ denote the
	associated set of locla distributions.  Let $X_1, \ldots, X_n$
	denote the associated $m$-local random variables.  The main result of
	this section shows that there exists a small set of seed
	vertices fixing whose value determines the value of almost
	every other vertex.  Formally, we show the following

\begin{lemma} \label{lem:existsgoodseedset}
For every integer $m$, there exists a subset of vertices $S \sse V$ of
size $|S| = k^2m$ such that
$$ \E_{i} \left[\Var[X_i|X_S]\right] \leq O\left(\frac{\eta}{\lambda_m}\right) $$
\end{lemma}

	To this end, we will relate conditioning a random
	variable $X_i$ on a set $X_S$, to projecting the SDP vectors
	corresponding to the variable $X_i$ in to the span of the
	vectors corresponding to $X_S$.  This analogy is formalized in
	the following lemma.

\begin{lemma} \label{lem:vectorstoconditionalvariance}
	Let $X_1, X_2, \ldots, X_r$ be random variables with range
	$[k]$ with a joint distribution $\mu$ associated with them.
	For each $i \in [r]$, $a \in [k]$, let $X_{ia}$ be the
	indicator of the event that $X_{i} = a$.
	Let us suppose there exists vectors $\set{v_{ia}}_{i \in [r], a \in
	[k]}$ such that
	\beq \iprod{v_{ia}, v_{jb}} = \Pr_{X_i X_j}\Set{X_i = a, X_j =
	b}   = \E[X_{ia} X_{jb}] \mper \eeq
	Then, for every subset $S \sse [r]$ we have (1) $\Var\left[
			X_{ia} | X_S \right] \leq   \norm{P_{S} v_{ia}}^2 \mper$
	 and (2) $ \Var\left[X_{i} |
		X_S\right] \leq  \sum_{a \in [k]} \norm{P_{S} v_{ia}}^2 \mper$
	where $P_S$ is the projector of $\R^n$ in to the space
	orthogonal to the span of $\set{v_{jb}}_{j \in S, b \in [k]}$.
\end{lemma}
\begin{proof}
	Let us suppose $v_{ia} = \sum_{j \in S, b \in [k]} c_{jb}
	v_{jb} + P_{S}v_{ia}$.  Define a random variable $C_S$ as
	follows,
	\beq C_S \defeq \sum_{j \in S, b \in [k]} c_{jb} X_{jb} .\eeq
	Note that on fixing the values $\{X_{j}\}_{j \in
	S}$, the random variable $C_S$ is fixed.

	By the definition of variance of a real random variable we have the
	following inequality.
	\begin{equation*}
		\Var\left[ (X_{ia} | x_{S})\right] =  \min_{C} \E[  (X_{ia} - C)^2
		| x_{S}] \leq \E_{X_{ia}|x_{S}}[ (X_{ia} - C_S)^2 | x_{S}]\mper
	\end{equation*}
	Averaging the above inequality over the settings of $x_{S}$,
	we get
	\begin{equation} \label{eq:projvariance1}
		\Var[X_{ia}|X_S] = \E_{x_{S}} \Var(\{X_{ia}|x_{S}\}) \leq
		\E_{x_{S}}\E_{X_{ia}|x_{S}} [(X_{ia} - C_S)^2 | x_{S}]
		= \E_{\mu}\left[(X_{ia} - \sum_{j \in S,b\in [k]}
		c_{jb}X_{jb})^2 \right] \mper
	\end{equation}
	Note that the second moments of the random variables
	$\{X_{ia}\}_{i \in [r], a \in [k]}$ match with the
	corresponding inner products of vectors $\{v_{ia}\}_{i \in
	[r],a \in[k]}$.  Hence,
	\begin{equation}\label{eq:projvariance2}
	\E_{\mu}\left[(X_{ia} - \sum_{j \in S,b\in [k]}
		c_{jb}X_{jb})^2 \right] =  \norm{ v_{ia} - \sum_{j \in S,b\in [k]}
		c_{jb}v_{jb}}^2  = \norm{P_{S} v_{ia}}^2  \mper
	\end{equation}
	The claim (1) follows from \eqref{eq:projvariance1} and
	\eqref{eq:projvariance2}.

	The claim (2) follows from
	(1) and the definition of variance
	of a random variable taking values over $[k]$.
\end{proof}

\begin{proof}[Proof of \pref{lem:existsgoodseedset}]

	For a subset $T \sse \mathcal{V} = \set{v_{ia}}_{i \in V, a
	\in [k]}$, let $S_T$ be the set of vertices associated with it
	namely,
	$$ S_{T} = \{i \in V \mid \exists b \in [k], v_{ib} \in
	T \} \mcom$$
	Let $Q_T$ denote the projector on to the subspace orthogonal
	to span of $\{v_{ia} | i \in S_{T}, a \in [k]\}$.
	In particular, $Q_T$ is a projector on to a subspace
	orthogonal to $T$ for all $T \sse \mathcal{V}$.
	
	Apply \pref{thm:low-rank-approx-general} on the set of vectors
	$\mathcal{V} = \set{v_{ia}}_{i \in V, a \in [k]}$ with the
	projectors $Q_T$ for a subset $T \sse \mathcal{V}$.
	\pref{thm:low-rank-approx-general} implies that there exists a
	choice of $T \sse \mathcal{V}$ of size $|T| = k^2m$ such that if $u_{ia} =
	Q_T v_{ia}$ then,
	\begin{equation}\label{eq:lowglobalinnerproducts}
	\E_{i,j,a,b} \norm{u_{ia}} \, \norm{u_{jb}} \,
	\iprod{\bar u_{ia},\bar u_{jb}}^2 \leq \frac{1}{k^2m}
	\end{equation}
	Let $S_T \sse V$ be the vertex set associated with the set of
	vectors $T$.  We will drop the subscript and refer to this set
	as $S$.  Let $P_S$
	denote the projector in to the space orthogonal to the span of
	vectors $\set{v_{ia}}_{i \in S, a \in [k]}$.

	Let us fix some notation: $u_{ia} \defeq P_S v_{ia}$ , $\tilde u_{ia} \defeq \norm{u_{ia}} \bar u_{ia}^{\otimes 2} \otimes \bar v_{ia}$.
	As for each $i \in V$, the vectors $\set{v_{ia}}_{a \in [k]}$ are orthogonal to each other, the set of vectors $\set{\tilde u_{ia}}_{a \in [k]}$ are orthogonal to each other too.   For each vertex $i \in V$, we can associate a vector $U_i$	defined as $U_i \defeq \sum_{a \in [k]} \tilde u_{ia}$.

	From
	\eqref{eq:lowglobalinnerproducts}, we get the following bound
	on the average correlation of vectors $\set{\tilde u_{ia}}_{i
	\in V, a \in [k]}$,
	\begin{equation}
	\E_{i,j,a,b} \iprod{ \tilde  u_{ia}, \tilde  u_{jb}} \leq
		\E_{i,j,a,b} \norm{u_{ia}} \, \norm{u_{jb}} \,
		\iprod{\bar u_{ia},\bar u_{jb}}^2 \leq \frac{1}{k^2
		m}
	\end{equation}
	Using the low global correlation between vectors
	 $\set{\tilde u_{ia}}_{i \in V, a \in [k]}$
	 (\eqref{eq:lowglobalinnerproducts}), we bound the global correlation between the vectors $\set{U_i}_{i \in
	V}$ as shown below,
 \[
		\E_{i,j \in V} \left[ \iprod{U_i, U_j} \right]
			= \E_{i,j \in V} \sum_{a,b \in [k]}
			\iprod{ \tilde
			u_{ia}, \tilde u_{jb}}
			 = k^2 \E_{i,j \in V, a,b \in [k]}
			\iprod{ \tilde	u_{ia}, \tilde u_{jb}}
			\leq \frac{1}{m}  \mper
	\]
	From \pref{lem:localtoglobal}, the low global correlation of
	vectors $\set{U_i}_{i \in V}$ implies that their squared
	length is small, i.e.,
	\beq \E_{i \in V} \left[\norm{U_i}^2\right] \leq
	O\left(\frac{\eta}{\lambda_m}\right) \mper
 \eeq
	Notice that
\beq
\norm{U_i}^2 = \sum_{a \in [k]} \norm{\tilde u_{ia}}^2 = \sum_{a \in [k]} \norm{P_s v_{ia}}^2 \mper
\eeq
 	By \pref{lem:vectorstoconditionalvariance}, this implies that
	\beq
		\E_{i \in V} \Var[X_i | X_S] \leq
		O\left(\frac{\eta}{\lambda_m}\right) \mper
	\eeq
\end{proof}

\begin{lemma}[High Local Correlation]\label{lem:highlocalcorrelation}
If $\mathcal{V}$ is an SDP solution to unique games with
	value $1 - \eta$, i.e.,
	\beq \E_{(i,j) \in E} \sum_{a \in [k]} \norm{v_{ia} -
          v_{j\pi_{ij}(a)}}^2 \leq \eta  \mcom~\eeq
	then
        \beq
        \E_{(i,j) \in E} \norm{U_i- U_j}^2 \leq 3\eta \mper
        \eeq
\end{lemma}
  We defer the proof to Appendix~\ref{app:proofs}.

\begin{lemma}[Local Correlation $\longrightarrow$ Global Correlation]
	\label{lem:localtoglobal}
	If the vectors $\set{U_i}_{i \in V}$ satisfy,
	\beq \E_{i} \left[\norm{U_i}^2\right] \geq
	\frac{4\eta}{\lambda_m} \mcom\eeq
	then the average correlation among the vectors $\set{U_i}_{i
	\in V}$ is at least $\nfrac{1}{m}$, i.e.,
	\beq \E_{i,j \in V} \iprod{U_i, U_j} \geq \frac{1}{m} \mper\eeq
\end{lemma}
\begin{proof}
	By \pref{lem:highlocalcorrelation}, the vectors $\set{U_i}$
	satisfy
	$$\E_{(i,j) \in E} \norm{U_i - U_j}^2 \leq 3\eta \mper$$
	This implies that,
	$$ \E_{(i,j) \in E} \iprod{U_i, U_j} \geq \E_{i} \norm{U_i}^2
	- \frac{3}{2}\eta \mper$$
	Let $\E_{i} \norm{U_i}^2 = C \geq 4\eta/\lambda_m$.  Normalize the vectors $U_i$ so as to make their average
	squared length equal to $1$. The resulting vectors have
	correlation at least $(1-\eta/C) \geq 1-\lambda_m/2$.  By
	\pref{lem:sdptothresholdrank}, this implies that $\E_{i,j \in
	V} \iprod{U_i, U_j}^2 \geq \frac{1}{m}$.  Since $\norm{U_i}
	\leq 1$ for all $i \in V$, we get
	\beq\E_{i,j \in
	V} \iprod{U_i, U_j} \geq  \E_{i,j \in	V} \iprod{U_i, U_j}
	\geq \frac{1}{m} \mper\eeq
\end{proof}

\subsection{Wrapping Up}

Our main result about \uniquegames (\pref{thm:subexp}) is  a direct consequence of \pref{thm:ug-higher-eigen} and \pref{thm:subexponential}  presented here.
\begin{theorem} \label{thm:ug-higher-eigen}
	For every positive integer $m$, there exists an algorithm
	running in time $n^{O(mk^2)}$ that given a unique games instance
	$\Gamma$ over alphabet $[k]$ with value $1-\eta$, finds a labelling satisfying
	$1-O(\frac{\eta}{\lambda_m})$ fraction of the edges. Here
	$\lambda_m$ is the $m^{th}$ smallest eigen value of the
	Laplacian of the constraint graph $\Gamma$.
\end{theorem}
\begin{proof}
	The algorithm proceeds by solving the $k^2m + 2$-round Lasserre
	SDP for the given instance.  Starting with the SDP solution,
	the algorithm runs the propagation rounding algorithm starting
	from every possible {\em seed set} $S$ of size $|S| = k^2m$.

	By \pref{lem:existsgoodseedset}, there exists one such set $S$
	for which we have,
	\begin{equation}\label{eq:lowvariance}
	\E_{i} \left[\Var[X_i|X_{S}]\right] \leq
	O\left(\frac{\eta}{\lambda_m}\right) \mper
	\end{equation}
	
	Let $\mu^{|S}$ denote the distribution over global assignments
	output by the propagation rounding scheme.  For an edge
	$(i,j)$, let $\mu_{ij}$ denote the local distribution over
	$[k]^2$ suggested by the SDP solution.  From
	\pref{lem:propagationroundingerror}, the statistical distance
	between $\mu_{ij}$ and $\mu^{|S}_{ij}$ is at most
	$$ \norm{\mu^{|S}_{ij} - \mu_{ij}}_1 \leq \Var[X_i | X_S] +
	\Var[X_j | X_S] \mper $$
	Hence for every edge $(i,j)$,
	$$ \Pr_{\mu^{|S}_{ij}}[x_i = \pi_{ij}(x_j)] \geq
	\Pr_{\mu_{ij}}[x_i = \pi_{ij}(x_j)] - \Var[X_i | X_S] -
	\Var[X_j | X_S] \mper$$
	Averaging over all the edges we see that,
	\begin{eqnarray*}
	 \E_{i,j \in E} \Pr_{\mu^{|S}_{ij}}[x_i = \pi_{ij}(x_j)]
	 & \geq & \E_{i,j \in E} \left[ \Pr_{\mu_{ij}}[x_i = \pi_{ij}(x_j)]
		\right]	- \E_{i,j \in E} \Var[X_i | X_S] - \E_{i,j \in E} \Var[X_j |
		X_S] \\
	 & \geq & \E_{i,j \in E} \left[ \Pr_{\mu_{ij}}[x_i = \pi_{ij}(x_j)]
		\right]	- 2 \E_{i\in V} \Var[X_i | X_S]  \\
		& \geq & \mathsf{Val}(\mathcal{V})	- 2 \E_{i\in V} \Var[X_i | X_S]
	\end{eqnarray*}
	where $\mathsf{Val}(\cV)$ is the SDP objective value of the
	solution $\mathcal{V}$.  Along with \eqref{eq:lowvariance},
	this implies that the algorithm on the choice of the
	appropriate seed set $S$ would find a solution with value at
	least $1 - \eta - O(\frac{\eta}{\lambda_m})$.
\end{proof}

\begin{theorem}\label{thm:subexponential}
  There exists an algorithm that given a \uniquegames instance $\Gamma$
  with vertex set $[n]$, label set $[k]$, and optimal value $1-\e$, finds
  an assignment with value at least $\half$ by rounding an
  $k^2\cdot n^{O(\e^{1/3})}$-round Lasserre solution.
\end{theorem}

\begin{proof}[Proof sketch] The proof follows by combining our
  propagation  rounding  and the decomposition
  theorem of \cite{AroraBS10}. The latter result allows us to partition the
  input graph into disjoint components each with $1-c\e$ rank at most
  $n^{O(\e^{1/3})}$ by removing at most $0.01$ fraction of the edges in our
  input graph. An SDP solution for the input graph induces a solution for
  each of the components, and hence we can round the solution for each
  component separately using propagation rounding.
\end{proof}

\section*{Conclusions}

We have shown that $n^{O(\e^{1/3})}$ rounds of an SDP hierarchy suffice for solving the \uniquegames problem on $(1-\e)$-satisfiable instances. The best lower bound known for the hierarchy we used is $\log\log^{\Omega(1)} n$~\cite{RaghavendraS09c,KhotS09}, and so a natural question, with obvious relevance to the unique games conjecture, is which bound is closer to the truth. The fact that our algorithm's running time for $r$ rounds is only $2^{O(r)}$ (as opposed to $n^{O(r)}$), challenges the interpretation of lower bounds in the range $[\omega(1),\O(\log n)]$ as corresponding to super-polynomial running time, and so provides further motivation to the question of whether the current hierarchy lower bounds can be improved further.

With the exception of the \smallsetexpansion problem, we do not know how to translate algorithms for \uniquegames into other computational problems. We hope that our ideas will help in combining the \cite{AroraBS10} subexponential algorithm for \uniquegames with SDP-based method to make progress on other \uniquegames-hard computational problems. Indeed, Arora and Ge (personal communication) recently used the ideas of this work to obtain improved algorithms for $3$-coloring on some interesting families of instances. A concrete open question along similar lines is whether one can get an algorithm  for the \maxcut problem with approximation factor  $\e$ better than the factor of the Goemans-Williamson algorithm that runs in time $\exp(n^{\poly(\e)})$.

For general 2-CSPs, we know that some instances will require a large number of hierarchy rounds, but it's interesting to see whether there is any clean characterization of the instances on which SDP hierarchies do well, encompassing, say, both low threshold rank graphs and planar graphs. Another interesting question is to find the right generalization of the low threshold rank condition to $k$-CSPs for $k>2$.

\addreferencesection
\bibliographystyle{amsalpha}
\bibliography{sdphierarchies}

\appendix

\section{Faster Algorithms for SDP hierarchies}

In this section, we argue that our rounding algorithm also works with weaker
SDP hierarchies.
We will show that for these weaker hierarchies, a near-optimal $m$-round
solution can be computed in time $2^{O(r)}\poly(n)$.
Due to the equivalence of optimization and separation, it is enough to
describe a separation oracle with running time $2^{O(r)}\poly(n)$.
Given a collection of vectors $\set{v_\ia}$, the separation oracle either
has to output a good assignment or it has to output a valid linear
constraint violated by the inner products of the input vectors.

We argue that such a separation oracle can easily be extracted from our
rounding algorithm.
Our rounding algorithm for \uniquegames first selects a set $S$ of roughly
$m$ vertices, then samples an assignment $x_S$ for these vertices, and
finally samples labels $x_i$ for the remaining vertices from the local
distributions conditioned on the event $x_S$.
The selection of the set $S$ depends only on the SDP vectors $\set{v_\ia}$
but not on the local distributions (which are not known to the separation
oracle).

Hence, given vectors $\set{v_{ia}}$, our separation oracle can simply work
as follows:
\begin{enumerate}
\item Select a vertex subset using \pref{thm:low-rank-approx} based on the
  given vectors $\set{v_\ia}$.
\item Using linear programming, find local distributions that are
  as consistent as possible with the inner products of the vectors $\set{v_\ia}$.
  If these local distributions match the inner products  sufficiently
  closely, then our propagation rounding algorithm will succeed.
  On the other hand, if the local distributions do not match the inner
  products closely enough, then we can find a valid linear constraints that is
  violated by the inner product of the given vectors. (This separating
  linear constraint can be obtained from the dual solution of the linear
  program that was used to find the best local distributions.)
\end{enumerate}

\section{Omitted proofs from Section~\ref*{sec:general-2-csp} and Section~\ref*{sec:round-uniq-games}} \label{app:proofs}

This appendix contains the proofs for some omitted proofs.

\begin{lemma}[High Local Correlation](\pref{lem:highlocalcorrelation}
	restated)
If $\mathcal{V}$ is an SDP solution to unique games with
	value $1 - \eta$, i.e.,
	$$	\E_{(i,j) \in E} \sum_{a \in [k]} \norm{v_{ia} -
		v_{j\pi_{ij}(a)}}^2 \leq \eta  \mcom $$
	then
	\begin{equation}
		\E_{(i,j) \in E} \norm{U_i- U_j}^2 \leq 3\eta \mcom
	\end{equation}

\end{lemma}

\begin{proof}
	Observe that the vectors $u_{ia}$ are projections of $v_{ia}$ and projections
	shrinks distances, which implies that the $\set{u_{ia}}$ vectors
        are correlated across         constraints of the \uniquegames instance,
	\begin{equation}
		\E_{(i,j) \in E} \sum_{a \in [k]} \norm{u_{ia} -
		u_{j\pi_{ij}(a)}}^2 \leq \E_{(i,j) \in E} \sum_{a \in [k]} \norm{v_{ia} -
		v_{j\pi_{ij}(a)}}^2 \leq \eta  \mper
	\end{equation}
	Let $\tilde u_{ia} = \norm{u_{ia}} \bar u_{ia}^{\otimes 2}\tensor
        \bar v_{ia}$.
        Notice that $\iprod{\tilde u_\ia,\tilde u_{ib}}=0$ for distinct
        labels $a,b\in[k]$.
        We claim that the $\tilde u_\ia$ vectors are also correlated across
        constraints of the \uniquegames instance,
	\begin{claim} \label{claim:tildegoodugsolution}
	\begin{equation}
		\E_{(i,j) \in E} \sum_{a \in [k]} \norm{\tilde u_{ia}
		- \tilde u_{j\pi_{ij}(a)}}^2 \leq 3 \eta
	\end{equation}
	\end{claim}

        \begin{proof}
          The following identity relates the distance of vectors to
          differences in their norms and the distance of the corresponding
          unit vectors,
          \begin{math}
            \snorm{x-y} = (\norm{x}-\norm{y})^2 + \norm{x} \, \norm{y}\,
            \snorm{\bar x-\bar y}\mper
          \end{math}
          Since $\norm{\tilde u_\ia}=\norm{u_\ia}$, we get
          \begin{displaymath}
            \norm{\tilde u_{ia}
		- \tilde u_{j\pi_{ij}(a)}}^2
              = \Bigparen{\norm{u_\ia}-\norm {u_{j\pi_{ij}(a)}}}^2
              + \norm{u_\ia}\,\norm {u_{j\pi_{ij}(a)}}\,
              \Snorm{\bar u_\ia^{\tensor 2} \tensor \bar v_\ia -
                \bar u_{j\pi_{ij}(a)}^{\tensor 2} \tensor \bar v_{j\pi_{ij}(a)}}
          \end{displaymath}
          Since $\snorm{\bar x_1 \tensor \bar x_2 -\bar y_1\tensor \bar
            y_2}\le \snorm{\bar x_1-\bar y_1}+\snorm{\bar x_2-\bar y_2}$,
          we can further upper bound
          \begin{align*}
            \norm{\tilde u_{ia} - \tilde u_{j\pi_{ij}(a)}}^2 &\le
            \Bigparen{\norm{u_\ia}-\norm {u_{j\pi_{ij}(a)}}}^2 +
            \norm{u_\ia}\,\norm {u_{j\pi_{ij}(a)}}\, \Paren{2\Snorm{\bar
                u_\ia- \bar u_{j\pi_{ij}(a)}}+ \Snorm{ \bar v_\ia - \bar
                v_{j\pi_{ij}(a)}}}
            \\
            &\le 2\snorm{u_\ia- u_{j\pi_{ij}(a)}} +\snorm{ v_\ia -
              v_{j\pi_{ij}(a)}}\mper
          \end{align*}
          (In the last step, we again used the identity \begin{math}
            \snorm{x-y} = (\norm{x}-\norm{y})^2 + \norm{x} \, \norm{y}\,
            \snorm{\bar x-\bar y}\mper
          \end{math}
          and the fact that $\norm{u_\ia}\,\norm {u_{j\pi_{ij}(a)}}\le
          \norm{v_\ia}\,\norm {v_{j\pi_{ij}(a)}}$.)
          By averaging over the label set and the edges of the graph, it
          follows as claimed that
          \begin{displaymath}
            \E_{ij\in E}\sum_{a\in[k]}
            \norm{\tilde u_{ia} - \tilde u_{j\pi_{ij}(a)}}^2
            \le 3\eta
            \mper\qedhere
          \end{displaymath}
        \end{proof}

        To finish the proof of the lemma, we relate the distances
        $\snorm{U_i-U_j}$ across an edge $ij\in E$ to distances of the
        vectors $\set{\tilde u_\ia}$ across the constraint $\pi_{ij}$,
        \begin{align*}
          \iprod{U_i,U_j}
          &= \sum_{a,b} \iprod{\tilde u_\ia,\tilde u_{jb}}\\
          &\ge \sum_{a} \iprod{\tilde u_{\ia},\tilde u_{j\pi_{ij}(a)}}
          \quad\text{(using non-negativity of involved inner products)}
          \\
          & =\sum_{a} \tfrac12\Paren{\snorm{\tilde u_{\ia}}+\snorm{\tilde
              u_{j\pi_{ij}(a)}}-\snorm{\tilde u_\ia - \tilde
              u_{j\pi_{ij}(a)}}}
          \\
          & = \tfrac12\snorm{U_i} + \tfrac12\snorm{U_j}
          -\sum_a \tfrac12\snorm{\tilde u_\ia - \tilde          u_{j\pi_{ij}(a)}}\mper
        \end{align*}
        Rearranging gives $\snorm{U_i-U_j}
        \le \sum_a \snorm{\tilde u_\ia - \tilde
          u_{j\pi_{ij}(a)}}$.
        Hence, using \pref{claim:tildegoodugsolution},
        \begin{displaymath}
          \E_{ij\in E}\snorm{U_i-U_j} \le \E_{ij\in E} \sum_a\snorm{\tilde u_\ia - \tilde
          u_{j\pi_{ij}(a)}}\le 3\eta.
        \end{displaymath}
   \end{proof}

\restatelemma{lem:ug-tensoring-trick}

\begin{proof}
  Let $\set{u_\ia }$ be the collection of vectors such that
   $\iprod{u_\ia,u_\jb} = \Cov(X_\ia,X_\jb)$.
  Note that $\snorm{u_\ia}=\Var X_\ia$.
  Let $v_\ia=u_\ia + \E X_\ia \,v_\eset$, where $v_\eset$ is a unit vector
  orthogonal to all vectors $u_\ia$.
  Define $\vec v_i = \sum_{a} \norm{u_\ia}{\bar u_\ia} ^{\tensor 2}\tensor
  {\bar v}_\ia^{\tensor 2}$.
  (Here, $\bar x$ denotes the unit vector in direction $x$.)
  Let us first lower bound the inner product of $\vec v_i$ and $\vec v_j$,
  \begin{align*}
    \Bigparen{\sum_{a\in [k]} \bigabs{\Cov(X_\ia,X_{j\,\pi(a)})}}^4
    & = \Bigparen{\sum_{a} \norm{u_\ia} \norm{u_{j\,\pi(a)}} \abs{\iprod{\bar u_\ia,\bar
          u_{j\,\pi(a)}}} }^4
    \\
    &=\Paren{
      \sum_{a} \norm{u_\ia}\, \norm{u_{j\,\pi(a)}} \cdot
      \abs{\iprod{\bar u_\ia,\bar  u_{j\,\pi(a)}}}
      \sqrt{\tfrac{\norm{u_\ia}\,\norm{u_{j\,\pi(a)}}}
        {\norm{v_\ia}\,\norm{v_{j\,\pi(a)}}}}\cdot
      \sqrt{\tfrac{\norm{v_\ia}\,\norm{v_{j\,\pi(a)}}}
        {\norm{u_\ia}\,\norm{u_{j\,\pi(a)}}}}
    }^4
    \\
    &\le \Paren{
      \sum_{a} \norm{u_\ia}\, \norm{u_{j\,\pi(a)}} \cdot
      \iprod{\bar u_\ia,\bar  u_{j\,\pi(a)}}^2
      \tfrac{\norm{u_\ia}\,\norm{u_{j\,\pi(a)}}}
        {\norm{v_\ia}\,\norm{v_{j\,\pi(a)}}}
    }^2
    \\
    &\quad\cdot
    \Paren{
      \sum_{a} \norm{u_\ia}\, \norm{u_{j\,\pi(a)}} \cdot
      \tfrac{\norm{v_\ia}\,\norm{v_{j\,\pi(a)}}}
        {\norm{u_\ia}\,\norm{u_{j\,\pi(a)}}}
    }^2
    \quad\using{Cauchy--Schwarz} \\
    &
    \le \Paren{\sum_{a} \norm{u_\ia}\, \norm{u_{j\,\pi(a)}} \cdot
      \abs{\iprod{\bar u_\ia,\bar  u_{j\,\pi(a)}}
      \iprod{\bar v_\ia,\bar v_{j\,\pi(a)}}}
    }^2
  \\
  &\qquad\qquad\using{$\iprod{v_\ia,v_{j\,\pi(a)}}\ge
    \iprod{u_\ia,u_{j\,\pi(a)}}$
    and $\tsum_a \norm{v_\ia}\,\norm{v_{j \,\pi(a)}}\le 1$}
  \\
  &  \le \sum_{a} \norm{u_\ia}\, \norm{u_{j\,\pi(a)}} \cdot
  \iprod{\bar u_\ia,\bar  u_{j\,\pi(a)}}^2
  \iprod{\bar v_\ia,\bar v_{j\,\pi(a)}}^2
  \\
  &\qquad\qquad\using{Cauchy--Schwarz and $\sum_a
    \norm{u_\ia}\,\norm{u_{j\,\pi(a)}}\le 1$}
  \\
  & \le \iprod{\vec v_i,\vec v_j}\mper
  \end{align*}
  On the other hand, we can upper bound the inner product of $\vec v_i$ and
  $\vec v_j$,
  \begin{align*}
    \iprod{\vec v_i,\vec v_j}%
    &=\sum_{a,b} \norm{u_\ia}\, \norm{u_{jb}} %
    \cdot \iprod{\bar u_\ia,\bar u_{jb}}^2 \iprod{\bar v_\ia,\bar v_{jb}}^2
    \\
    &\le \sum_{a,b} \tfrac 12 (\snorm{u_\ia}+ \snorm{u_{jb}}) %
    \cdot \iprod{\bar u_\ia,\bar  u_{jb}}^2 \\
    & = \sum_{a,b}\tfrac 12 (\tfrac 1{\Var X_\ia} +\tfrac 1{\Var X_\jb}) %
    \,\Cov(X_\ia,X_\jb) %
    \mper
  \end{align*}
  Finally, the vectors $\vec v_1\ldots,\vec v_n$ are in the unit ball,
  \begin{displaymath}
    \snorm{\vec v_i}
    = \sum_{a,b} \norm{u_\ia}\, \norm{u_{ib}} %
    \cdot \iprod{\bar u_\ia,\bar u_{ib}}^2 \iprod{\bar v_\ia,\bar v_{ib}}^2
    = \sum_a \snorm{u_\ia} \le 1\mper
  \end{displaymath}
  Here, we are using the fact that $\iprod{v_\ia,v_{ib}}=0$ for all
  distinct $a,b\in [k]$.
\end{proof}

\section{Facts about Variance}
\label{sec:facts-about-variance}

\begin{lemma}
  Let $X$ and $Y$ be jointly-distributed random variables.
  Assume that $Y$ has finite range.
  %
  Let $Z$ be the orthogonal projection of the random variable $X$ onto the
  subspace of functions of the random variable $Y$.
  Then,
  \begin{displaymath}
    \E_{\set Y}\Var\Brac{X \Mid Y} = \E X^2 - \E Z^2\mper
  \end{displaymath}
\end{lemma}

\begin{proof}
  By construction $Z$ is a function $f(Y)$ of the random variable $Y$ and
  $X-Z$ is orthogonal to all functions of the variable $Y$.
  Hence, $\E\brac{ X \mid Y = y}= f(y)$.
  Therefore, the expected variance of $\brac{X\mid Y}$ is
  \begin{align*}
    \E_{\set{Y}} \Var\Brac{X \Mid Y} %
    & = \E X^2 - \E_{\set Y} \Bigparen{\E\brac{X\mid Y}}^2
    \\
    & =\E X^2 - \E_{\set Y} f(Y)
  \end{align*}
  which gives the desired identity using $Z=f(Y)$.
\end{proof}

\begin{lemma}\label{lem:condition-indicator}
  Let $X$ and $Y$ be as in the previous lemma.
  Suppose the range of $Y$ has cardinality $2$.
  Then,
  \begin{displaymath}
    \E_{\set Y} \Var\brac{X \mid Y} %
    = \Var X - \Cov(X,Y)^2/\Var(Y)
    \mper
  \end{displaymath}
\end{lemma}

\begin{proof}
  Without loss of generality, we may assume that $\E X = \E Y=0$ and $\E Y^2=1$.
  Then, the set of random variables $\set{1,Y}$ is an orthonormal basis for
  the subspace of functions of~$Y$.
  Let $\rho = \E X Y$.
  Then, $\rho Y$ is the orthogonal projection of $X$ to the subspace of
  function of $Y$.
  (Here, we use the assumption $\E X= 0$.)
  Hence, using the previous lemma,
  \begin{displaymath}
    \E_{\set Y} \Var\brac{X \mid Y}
    = \E X^2 - \E (\rho Y)^2 = \E X^2 - \rho^2\mcom
  \end{displaymath}
  which is the desired identity because $\E X^2=\Var X$ and $\rho^2 =
  \Cov(X,Y)^2/\Var Y$.
\end{proof}

\end{document}

%% file: macros.tex
\usepackage{etex}

\usepackage[l2tabu, orthodox]{nag}

\usepackage{xspace,enumerate}

\usepackage[dvipsnames]{xcolor}

\usepackage[T1]{fontenc}
\usepackage[full]{textcomp}

\usepackage[american]{babel}

\usepackage{mathtools}

\usepackage{bm}

\usepackage{amsthm}

\newtheorem{theorem}{Theorem}[section]
\newtheorem*{theorem*}{Theorem}

\newtheorem{claim}[theorem]{Claim}
\newtheorem*{claim*}{Claim}

\newtheorem*{proposition*}{Proposition}
\newtheorem{lemma}[theorem]{Lemma}
\newtheorem*{lemma*}{Lemma}

\newtheorem*{conjecture*}{Conjecture}

\newtheorem*{fact*}{Fact}

\newtheorem*{hypothesis*}{Hypothesis}

\theoremstyle{definition}
\newtheorem{definition}[theorem]{Definition}

\newtheorem{algorithm}[theorem]{Algorithm}

\newtheorem{remark}[theorem]{Remark}

\usepackage[letterpaper,
top=1in,
bottom=1in,
left=1.5in,
right=1.5in]{geometry}

\usepackage[varg]{txfonts} 
\renewcommand{\mathbb}{\varmathbb}

\ifnum\showkeys=1
\usepackage[color]{showkeys}
\fi

\ifnum\showcolorlinks=1
\usepackage[
pagebackref,
letterpaper=true,
colorlinks=true,
urlcolor=blue,
linkcolor=blue,
citecolor=OliveGreen,
]{hyperref}
\fi

\ifnum\showcolorlinks=0
\usepackage[
pagebackref,
letterpaper=true,
colorlinks=false,
pdfborder={0 0 0}
]{hyperref}
\fi

\usepackage{prettyref}

\newcommand{\savehyperref}[2]{\texorpdfstring{\hyperref[#1]{#2}}{#2}}

\newrefformat{eq}{\savehyperref{#1}{\textup{(\ref*{#1})}}}
\newrefformat{lem}{\savehyperref{#1}{Lemma~\ref*{#1}}}
\newrefformat{def}{\savehyperref{#1}{Definition~\ref*{#1}}}
\newrefformat{thm}{\savehyperref{#1}{Theorem~\ref*{#1}}}
\newrefformat{cor}{\savehyperref{#1}{Corollary~\ref*{#1}}}
\newrefformat{cha}{\savehyperref{#1}{Chapter~\ref*{#1}}}
\newrefformat{sec}{\savehyperref{#1}{Section~\ref*{#1}}}
\newrefformat{app}{\savehyperref{#1}{Appendix~\ref*{#1}}}
\newrefformat{tab}{\savehyperref{#1}{Table~\ref*{#1}}}
\newrefformat{fig}{\savehyperref{#1}{Figure~\ref*{#1}}}
\newrefformat{hyp}{\savehyperref{#1}{Hypothesis~\ref*{#1}}}
\newrefformat{alg}{\savehyperref{#1}{Algorithm~\ref*{#1}}}
\newrefformat{rem}{\savehyperref{#1}{Remark~\ref*{#1}}}
\newrefformat{item}{\savehyperref{#1}{Item~\ref*{#1}}}
\newrefformat{step}{\savehyperref{#1}{step~\ref*{#1}}}
\newrefformat{conj}{\savehyperref{#1}{Conjecture~\ref*{#1}}}
\newrefformat{fact}{\savehyperref{#1}{Fact~\ref*{#1}}}
\newrefformat{prop}{\savehyperref{#1}{Proposition~\ref*{#1}}}
\newrefformat{prob}{\savehyperref{#1}{Problem~\ref*{#1}}}
\newrefformat{claim}{\savehyperref{#1}{Claim~\ref*{#1}}}
\newrefformat{relax}{\savehyperref{#1}{Relaxation~\ref*{#1}}}
\newrefformat{red}{\savehyperref{#1}{Reduction~\ref*{#1}}}
\newrefformat{part}{\savehyperref{#1}{Part~\ref*{#1}}}

\newcommand{\Sref}[1]{\hyperref[#1]{\S\ref*{#1}}}

\usepackage{nicefrac}

\let\nfrac=\nicefrac

\newcommand{\half}{\nicefrac12}

\ifnum\usemicrotype=1
\usepackage{microtype}
\fi

\ifnum\showauthornotes=1
\newcommand{\Authornote}[2]{{\sffamily\small\color{red}{[#1: #2]}}}
\newcommand{\Authorcomment}[2]{{\sffamily\small\color{gray}{[#1: #2]}}}
\newcommand{\Authorstartcomment}[1]{\sffamily\small\color{gray}[#1: }

\newcommand{\Authorfnote}[2]{\footnote{\color{red}{#1: #2}}}
\newcommand{\Authorfixme}[1]{\Authornote{#1}{\textbf{??}}}
\newcommand{\Authormarginmark}[1]{\marginpar{\textcolor{red}{\fbox{\Large #1:!}}}}
\else
\newcommand{\Authornote}[2]{}
\newcommand{\Authorcomment}[2]{}
\newcommand{\Authorstartcomment}[1]{}

\newcommand{\Authorfnote}[2]{}
\newcommand{\Authorfixme}[1]{}
\newcommand{\Authormarginmark}[1]{}
\fi

\ifnum\showfixme=0

\fi

\usepackage{boxedminipage}

\newenvironment{mybox}
{\center \noindent\begin{boxedminipage}{1.0\linewidth}}
{\end{boxedminipage}
\noindent
}

\newcommand{\paren}[1]{(#1)}
\newcommand{\Paren}[1]{\left(#1\right)}

\newcommand{\Bigparen}[1]{\Big(#1\Big)}
\newcommand{\brac}[1]{[#1]}
\newcommand{\Brac}[1]{\left[#1\right]}

\newcommand{\Bigbrac}[1]{\Big[#1\Big]}
\newcommand{\abs}[1]{\lvert#1\rvert}

\newcommand{\bigabs}[1]{\big\lvert#1\big\rvert}

\newcommand{\card}[1]{\lvert#1\rvert}

\newcommand{\set}[1]{\{#1\}}
\newcommand{\Set}[1]{\left\{#1\right\}}

\newcommand{\norm}[1]{\lVert#1\rVert}
\newcommand{\Norm}[1]{\left\lVert#1\right\rVert}

\newcommand{\Bignorm}[1]{\Big\lVert#1\Big\rVert}
\newcommand{\normt}[1]{\norm{#1}_2}

\newcommand{\snorm}[1]{\norm{#1}^2}
\newcommand{\Snorm}[1]{\Norm{#1}^2}

\newcommand{\normo}[1]{\norm{#1}_1}
\newcommand{\Normo}[1]{\Norm{#1}_1}

\newcommand{\Bignormo}[1]{\Bignorm{#1}_1}

\newcommand{\iprod}[1]{\langle#1\rangle}

\newcommand{\Esymb}{\mathbb{E}}
\newcommand{\Psymb}{\mathbb{P}}
\newcommand{\Vsymb}{\mathbb{V}}

\DeclareMathOperator*{\E}{\Esymb}
\DeclareMathOperator*{\Var}{\Vsymb}
\DeclareMathOperator*{\ProbOp}{\Psymb}

\renewcommand{\Pr}{\ProbOp}

\newcommand{\Prob}[2][]{\Pr_{{#1}}\Set{#2}}

\usepackage{dsfont}
\usepackage{yfonts}
\usepackage{mathrsfs}

\newcommand{\tensor}{\otimes}

\newcommand{\textparen}[1]{\text{(#1)}}
\newcommand{\using}[1]{\textparen{using #1}}
\ifx\because\undefined
\newcommand{\because}[1]{\textparen{because #1}}
\else
\renewcommand{\because}[1]{\textparen{because #1}}
\fi

\renewcommand{\ij}{{ij}}

\newcommand{\bits}{\{0,1\}}

\newcommand{\super}[2]{#1^{\paren{#2}}}

\newcommand{\sm}{\setminus}

\newcommand{\defeq}{\stackrel{\mathrm{def}}=}

\newcommand{\seteq}{\mathrel{\mathop:}=}

\newcommand{\from}{\colon}

\newcommand{\Mid}{\;\middle\vert\;}

\renewcommand{\vec}[1]{{\bm{#1}}}

\newcommand{\mper}{\,.}
\newcommand{\mcom}{\,,}

\newcommand\bdot\bullet

\ifx\mathds\undefined 

\else

\fi

\DeclareMathOperator{\Tr}{Tr}

\DeclareMathOperator{\val}{val}

\DeclareMathOperator{\poly}{poly}

\DeclareMathOperator{\tsum}{{\textstyle \sum}}

\newcommand{\iid}{i.i.d.\xspace}

\newcommand{\Lovasz}{Lov\'asz\xspace}

\newcommand{\R}{\mathbb R}

\newcommand{\problemmacro}[1]{\texorpdfstring{\textsc{#1}}{#1}\xspace}

\newcommand{\uniquegames}{\problemmacro{Unique Games}}
\newcommand{\maxcut}{\problemmacro{Max Cut}}

\newcommand{\labelcover}{\problemmacro{Label Cover}}

\newcommand{\sparsestcut}{\problemmacro{Sparsest Cut}}

\newcommand{\smallsetexpansion}{\problemmacro{Small-Set Expansion}}

\newcommand{\cV}{\mathcal V}

\renewcommand{\leq}{\leqslant}
\renewcommand{\le}{\leqslant}
\renewcommand{\geq}{\geqslant}
\renewcommand{\ge}{\geqslant}

\ifnum\showdraftbox=1
\newcommand{\draftbox}{\begin{center}
  \fbox{%
    \begin{minipage}{2in}%
      \begin{center}%
          \Large\textsc{Working Draft}\\%
        Please do not distribute%
      \end{center}%
    \end{minipage}%
  }%
\end{center}
\vspace{0.2cm}}
\else
\newcommand{\draftbox}{}
\fi

\let\epsilon=\varepsilon

\numberwithin{equation}{section}

\newcommand{\MYstore}[2]{%
  \global\expandafter \def \csname MYMEMORY #1 \endcsname{#2}%
}

\newcommand{\MYload}[1]{%
  \csname MYMEMORY #1 \endcsname%
}

\newcommand{\MYnewlabel}[1]{%
  \newcommand\MYcurrentlabel{#1}%
  \MYoldlabel{#1}%
}

\newcommand{\MYdummylabel}[1]{}

\newcommand{\torestate}[1]{%
  \let\MYoldlabel\label%
  \let\label\MYnewlabel%
  #1%
  \MYstore{\MYcurrentlabel}{#1}%
  \let\label\MYoldlabel%
}

\newcommand{\restatetheorem}[1]{%
  \let\MYoldlabel\label
  \let\label\MYdummylabel
  \begin{theorem*}[Restatement of \prettyref{#1}]
    \MYload{#1}
  \end{theorem*}
  \let\label\MYoldlabel
}

\newcommand{\restatelemma}[1]{%
  \let\MYoldlabel\label
  \let\label\MYdummylabel
  \begin{lemma*}[Restatement of \prettyref{#1}]
    \MYload{#1}
  \end{lemma*}
  \let\label\MYoldlabel
}

\newcommand{\restateprop}[1]{%
  \let\MYoldlabel\label
  \let\label\MYdummylabel
  \begin{proposition*}[Restatement of \prettyref{#1}]
    \MYload{#1}
  \end{proposition*}
  \let\label\MYoldlabel
}

\newcommand{\restatefact}[1]{%
  \let\MYoldlabel\label
  \let\label\MYdummylabel
  \begin{fact*}[Restatement of \prettyref{#1}]
    \MYload{#1}
  \end{fact*}
  \let\label\MYoldlabel
}

\newcommand{\restate}[1]{%
  \let\MYoldlabel\label
  \let\label\MYdummylabel
  \MYload{#1}
  \let\label\MYoldlabel
}

\newcommand{\addreferencesection}{
  \phantomsection
  \addcontentsline{toc}{section}{References}
}

\newcommand{\sse}{\subseteq}

\newcommand{\e}{\epsilon}
\newcommand{\eps}{\epsilon}
\newcommand{\eset}{\emptyset}

\let\origparagraph\paragraph
\renewcommand{\paragraph}[1]{\origparagraph{#1.}}

\allowdisplaybreaks

\newcommand{\cclassmacro}[1]{\texorpdfstring{\textbf{#1}}{#1}\xspace}
\newcommand{\p}{\cclassmacro{P}}
\newcommand{\np}{\cclassmacro{NP}}

\ifnum\full=1
\newcommand{\beq}{\begin{displaymath}}
\newcommand{\eeq}{\end{displaymath}}
\else
\newcommand{\beq}{\begin{math}}
\newcommand{\eeq}{\end{math}}
\fi

\ifnum\full=1
\newcommand{\nnspace}{}
\else
\newcommand{\nnspace}{\vspace{-1ex}}
\fi